\documentclass{amsart}

\usepackage{amscd}
\usepackage{bbm}
\usepackage{mathrsfs}
\usepackage{eucal}
\usepackage{faktor}
\usepackage{xfrac}
\usepackage{braket}
\usepackage{relsize} 
\usepackage{nccmath}
\usepackage{mathtools}
\usepackage{comment}
\usepackage{verbatim}
\usepackage{wasysym}
\usepackage{graphicx}
\usepackage{caption}
\usepackage{import}
\usepackage{latexsym}
\usepackage{amsfonts}
\usepackage{amssymb}
\usepackage{amsmath}
\usepackage{mathtext}
\usepackage{cite}
\usepackage{enumerate}
\usepackage{float}
\usepackage{amsthm}
\usepackage{upgreek}
\usepackage{nicefrac}
\usepackage{multicol}
\usepackage{rotating}
\usepackage{array}

\usepackage{etoolbox}
\makeatletter
\pretocmd{\chapter}{\addtocontents{toc}{\protect\addvspace{15\p@}}}{}{}
\pretocmd{\section}{\addtocontents{toc}{\protect\addvspace{5\p@}}}{}{}
\makeatother
\let\oldtocsection=\tocsection
\let\oldtocsubsection=\tocsubsection
\let\oldtocsubsubsection=\tocsubsubsection
\renewcommand{\tocsection}[2]{\hspace{0em}\oldtocsection{#1}{#2}}
\renewcommand{\tocsubsection}[2]{\hspace{1.8em}\oldtocsubsection{#1}{#2}}
\renewcommand{\tocsubsubsection}[2]{\hspace{4.4em}\oldtocsubsubsection{#1}{#2}}

\usepackage[svgnames]{xcolor}
\definecolor{linkcolor}{HTML}{e88d67} 
\definecolor{citecolor}{HTML}{e88d67} 
\definecolor{urlcolor}{HTML}{e88d67} 
\definecolor{myNewColorA}{HTML}{fec3a6}
\definecolor{myNewColorB}{HTML}{ffaf80}
\definecolor{myNewColorC}{HTML}{fb8f67}
\definecolor{seagreen}{HTML}{337180}
\definecolor{mseagreen}{HTML}{369673}
\definecolor{darksalmon}{HTML}{e88d67}
\definecolor{silver}{HTML}{bbbbbb}
\definecolor{flowerblue}{HTML}{4e77fc}
\definecolor{tomato}{HTML}{ff6347}
\definecolor{orange}{HTML}{f2b13d}
\definecolor{darkgray}{HTML}{939393}

\usepackage[
  bookmarks = true,
  bookmarksopen = true,
  colorlinks = , 
  urlcolor = urlcolor, 
  linkcolor = linkcolor, 
  citecolor = citecolor, 
  linktoc = all,
  pagebackref = true
]{hyperref} 
\mathtoolsset{showonlyrefs,showmanualtags}

\DeclareMathOperator{\PIV}{PIV}
\DeclareMathOperator{\PPII}{PII}

\newcommand{\brackets}[1]{\left( #1 \right)}
\newcommand{\PoissonBrackets}[1]{\left\{ #1 \right\}}
\newcommand{\LieBrackets}[1]{\left[ #1 \right]}
\newcommand{\angleBrackets}[1]{\langle #1 \rangle}

\newcommand{\Painleve}{Painlev{\'e} }
\newcommand{\PPainleve}{Painlev{\'e}}
\newcommand{\Backlund}{B{\"a}cklund }

\newcommand{\PIVn}[1]{\text{P}_4^{ #1 }}

\theoremstyle{plain}
\newtheorem{thm}{Theorem}[]

\newtheorem{prop}{Proposition}[]

\theoremstyle{definition}

\newtheorem{exmp}{Example}[]

\theoremstyle{remark}
\newtheorem{rem}{Remark}

\usepackage{chngcntr}
\counterwithin*{thm}{section} 
\counterwithin*{lem}{section} 
\counterwithin*{claim}{section} 
\counterwithin*{prop}{section} 
\counterwithin*{corl}{section} 
\counterwithin*{defn}{section} 
\counterwithin*{exmp}{section} 
\counterwithin*{ex}{section} 
\counterwithin*{rem}{section} 

\usepackage{apptools}
\AtAppendix{\counterwithin{thm}{section}}
\AtAppendix{\counterwithin{lem}{section}}
\AtAppendix{\counterwithin{claim}{section}}
\AtAppendix{\counterwithin{prop}{section}}
\AtAppendix{\counterwithin{corl}{section}}
\AtAppendix{\counterwithin{defn}{section}}
\AtAppendix{\counterwithin{exmp}{section}}
\AtAppendix{\counterwithin{ex}{section}}
\AtAppendix{\counterwithin{rem}{section}}

\usepackage{tikz}

\begin{document}

    \title[A fully noncommutative analog of the $\PIV$ equation]{A fully noncommutative analog of the Painlev{\'e} IV equation and~a~structure~of its solutions}
    
    \author{Irina Bobrova}
    \noindent\address{\noindent Faculty of Mathematics, National Research University Higher School of Economics, Usacheva str. 6, Moscow,
    119048, Russia}
    \email{ia.bobrova94@gmail.com}
    
    \author{Vladimir Retakh}
    \noindent\address{\noindent 
    Department of Mathematics,
    Rutgers University,
    Piscataway, New Jersey, 08854, USA
    }
    \email{vretakh@math.rutgers.edu}
    
    \author{Vladimir Rubtsov}
    \noindent\address{\noindent 
    Maths Department, University of Angers, 
    Building I,
    Lavoisier Boulevard,
    Angers, 49045, CEDEX 01, France
    }
    \email{volodya@univ-angers.fr}
    
    \author{Georgy Sharygin}
    \noindent\address{\noindent 
    Department of Mathematics and Mechanics,
    Moscow State (Lomonosov) University,
    Leninskie Gory, d. 1,
    Moscow, 119991, Russia
    }
    \email{sharygin@itep.ru}

    \subjclass{Primary 37K20. Secondary 16B99, 16W25} 
    \keywords{
    noncommutative Toda equations, 
    noncommutative Painlev{\'e} equations, quasideterminants}
    
    \maketitle 

    \begin{flushright}
    \textit{
    Dedicated to the 70th birthday of V.V. Sokolov
    }
    \end{flushright}

    \begin{abstract}
    We study a fully noncommutative generalisation of the commutative fourth \Painleve equation that possesses solutions in terms of an infinite Toda system over an associative unital division ring equipped by a derivation.
    \end{abstract}
    
    
    \section*{Introduction}

The main purpose of the present paper is to illustrate the ``pure noncommutative'' approach to the theory of integrable systems by applying it to the noncommutative \Painleve equations, more precisely to the $\PIV$ equation and the related $\tau$-functions. This approach was first introduced in the paper \cite{Retakh_Rubtsov_2010} and can be regarded as a universal approach to the quantization or deformation of the well-known classical objects. In the cited paper the authors constructed an integrable fully noncommutative analog of the second \Painleve equation and gave an explicit description of its solutions in terms of Hankel quasideterminants. They also generalized results related to the solutions of the fully noncommutative Toda chain, suggested in papers \cite{etingof1997factorization}, \cite{gel1992theory}. Different versions of the noncommutative Toda equations can also be found in \cite{mikhailov1981reduction}, \cite{krichever1981periodic}. We will apply these results to construct a fully noncommutative version of the fourth \Painleve equation that admits solutions in terms of the infinite Toda system. The subject of our study is motivated by papers \cite{joshi2004generating}, \cite{joshi2006generating}, where the authors have explained the nature of the determinant solutions for the commutative $\PPII$ and $\PIV$ equations.

Consider an associative unital division ring $R$ over a field $F$. Let $D:R\to R$ be a derivation of $R$, i.e. an $F$-linear map that satisfies the Leibniz rule; for any $f \in R$ we put $D \, f = f'$. Below we will often refer to the elements of the ring $R$ as to \textit{functions}. We will fix an element $t\in R$ such that $t' = 1$ (so we assume that the differential equation $f'=1$ in $R$ has solutions) and for any scalar parameter $\alpha \in F$ we have $\alpha' = 0$. In the paper \cite{Retakh_Rubtsov_2010} the authors constructed solutions of the \textit{infinite Toda system} in terms of quasideterminants of Hankel matrices over $R$ (see Theorem 2.1 in \cite{Retakh_Rubtsov_2010}). This system contains two parts, ``positive'' and ``negative'', and can be written as
\begin{align}
    \label{eq:ncToda_thetan_0}
    (\theta_n' \theta_n^{-1})'
    &= \theta_{n + 1} \theta_n^{-1}
    - \theta_n \theta_{n - 1}^{-1},
    &
    n 
    &\geq 0,
    \\[2mm]
    \label{eq:ncToda_etam_0}
    (\eta_{m}^{-1} \eta_{m}') '
    &= \eta_{m}^{-1} \eta_{m - 1} - \eta_{m + 1}^{-1} \eta_{m},
    & 
    m &\leq 0,
\end{align}
here $\theta_1  = \eta_0^{-1}= \kappa_{1}$ and $\theta_0  = \eta_{-1}^{-1}= \kappa_{-1}^{-1}$ for some generic initial functions $\kappa_{-1}$ and $\kappa_1$. It turns out that if we impose some conditions on the functions $\kappa_{-1}$ and $\kappa_1$, then one can use the solutions of the noncommutative Toda equations to construct the solutions of \textit{the fully noncommutative \Painleve II equation} of the form
\begin{equation}
    \label{eq:ncP2}
    \tag*{$\PPII [u; \beta]$}
    u''
    = 2 u^3 - 2 t u - 2 u t + 4 \brackets{\beta + \tfrac12},
\end{equation}
where $u$, $t \in R$, $t' = 1$, and $\beta$ is a scalar parameter, i.e. $\beta\in F$ (in particular $\beta' = 0$). Note that the r.h.s. of the \ref{eq:ncP2} equation is written in symmetric or anticommutator form; this form is useful to construct some generalizations of the commutative $\PPII$ equation. Equation \ref{eq:ncP2} is a generalization of the matrix \Painleve II equation, obtained in \cite{Balandin_Sokolov_1998}, \cite{Adler_Sokolov_2020_1}, and of the quantum \Painleve II equation, suggested in \cite{nagoya2008quantum}, since in contrast with these two examples, there are no additional assumptions for the algebra $R$ in it.

In this paper we go further in this direction and suggest a fully noncommutative version of the commutative \Painleve IV equation
\begin{align}
    y ''
    &= \tfrac12 y^{-1} (y')^2
    + \tfrac32 y^3 
    - 2 t y^2
    + \brackets{\tfrac12 t^2 - (\alpha_1 - \alpha_0)} y
    - \tfrac12 \alpha_2^2 y^{-1},
\end{align}
where $y = y(t)$ and the scalar parameters $\alpha_0,\alpha_1,\alpha_2$ are such that $\alpha_0 + \alpha_1 + \alpha_2 = 1$. The solutions of this equation are expressed in terms of the solutions of noncommutative Toda equations under a certain ansatz for the functions $\kappa_{-1}$ and $\kappa_1$. This analog reads as 
\begin{align}
    \label{eq:ncP4sym0}
    &\left\{
    \begin{array}{lcl}
         f_{0}'
         &=& f_{0} f_{1}
         - f_{2} f_{0}
         + \alpha_0,  
         \\[2mm]
         f_{1}'
         &=& f_{1} f_{2}
         - f_{0} f_{1}
         + \alpha_1,
         \\[2mm]
         f_{2}'
         &=& f_{2} f_{0}
         - f_{1} f_{2}
         + \alpha_2,
    \end{array}
    \right.
\end{align}
where again $\alpha_0 + \alpha_1 + \alpha_2 = 1$; this system can be regarded as \textit{a fully noncommutative generalization of the $\PIV$ symmetric system} (in the commutative case system \eqref{eq:ncP4sym0} defines 3-periodic solutions of the dressing chain \cite{veselov1993}). This system admits the same \Backlund transformations as in the commutative case (see Table~\ref{tab:BTonP4sym}) and has an isomonodromic Lax representation, presented in Section \ref{sec:isompair}, that is equivalent to the Noumi-Yamada pair for the commutative $\PIV$ symmetric form \cite{noumi2000affine}. 
To construct solutions of \eqref{eq:ncP4sym0} in terms of the noncommutative Toda equations we use noncommutative analogs of the translation operators (for more details see Section \ref{sec:ncP4}). Therefore, the existence of \Backlund transformations, compositions of which form an affine Weyl group of type $A_2^{(1)}$, is significant for determinig solutions of \eqref{eq:ncP4sym0}. Unlike the \ref{eq:ncP2} equation, system \eqref{eq:ncP4sym0} is not written in the symmetric form. It means that a noncommutative analog of the $\PIV$ equation cannot be obtained by using the Weyl ordering.

We also remark that the order of the system can be reduced by the first integral
\begin{align}
    I
    &= f_0 + f_1 + f_2 - t,
\end{align}
but in noncommutative situation the resulting system cannot be written as an ODE of the second order, since to do this we have to invert an operator of the form $f \, g + g \, f$. 

Observe that unlike in the papers \cite{nagoya2004quantum} and \cite{nagoya2008quantum}, where quantum versions of the \Painleve equations were considered, we do not impose any additional conditions on $R$. So, system \eqref{eq:ncP4sym0} is a generalization of the quantum version of the $\PIV$ equation to a ``pure noncommutative'' case.

On the other hand, with the help of the matrix generalization of the \PPainleve-Kovalevskaya test, authors of the papers \cite{Bobrova_Sokolov_2021_1}, \cite{Bobrova_Sokolov_2021_2} have derived three non-equivalent matrix \Painleve IV systems, labeled $\PIVn{i}$, $i = 0, 1, 2$. The $\PIVn{0}$ system with scalar parameters is equivalent to \eqref{eq:ncP4sym0} with central variable $t$. There is also a fully noncommutative version of the $\PIVn{0}$ system:
\begin{equation} \label{roubtsov}
    \left\{
    \begin{array}{lcl} 
    u'
    &=& - u^2 + u v + v u + (k - 2) \, \bar zu - k\, u \bar z + \gamma_1,
    \\[2mm]
    v'
    &=& - v^2 + v u + u v + k\, \bar z v - (k - 2)\, v \bar z + \gamma_2,
    \end{array}
    \right.
\end{equation}
where $\bar z' = 1$ and $k \in  \mathbb{C}$; this system admits an isomonodromic Lax pair \cite{Bobrova_Sokolov_2021_2}. It turns out that solutions of this system can be expressed via solutions of the semi-infinite noncommutative Toda equations when $k = 0$ or $k = 2$. The noncommutative version \eqref{eq:ncP4sym0} allows us to express its solutions in both directions.

Up to the authors' knowledge, particular cases of the other two systems from the list in \cite{Bobrova_Sokolov_2021_1}, $\PIVn{1}$ and $\PIVn{2}$, do not have \Backlund transformations that define an affine Weyl group of type $A_2^{(1)}$\footnote{But they do possess a \Backlund transformation (see \cite{adler2020}). We are grateful to V.~E.~Adler, who drew our attention to~this~fact.}. Therefore they cannot be solved by the noncommutative Toda equations. The same remark holds for noncommutative systems of the $\PIV$ type, obtained in the paper \cite{cafasso2014non} by using the Riemann-Hilbert problem.

\subsection*{Organization of the paper} In Section \ref{sec:Todaeqs}, we give a short overview of the known facts about Toda equations and their solutions in the commutative (Section \ref{sec:scalToda}) and noncommutative cases (Section \ref{sec:ncToda}). The overview is based on the papers \cite{kajiwara1999determinant} and \cite{Retakh_Rubtsov_2010}. In Section \ref{sec:scalP4} we recall the Hirota bilinear form of the commutative $\PIV$ equation, that is derived by \Backlund transformations of the commutative $\PIV$ symmetric form. We note that, unlike the paper \cite{joshi2006generating}, we present a direct proof of the fact that solutions of the \Painleve IV equation can be constructed by solutions of the Toda equations under some assumptions for the general initial functions $\kappa_{-1}$ and $\kappa_1$ (see Theorem \ref{thm:scalP4solmainthm} in Section \ref{sec:scalHankel}). In Section \ref{sec:ncP4} we give the generalizations of the results from Section \ref{sec:scalP4} to a ``pure noncommutative'' case. Namely, we prove that solutions of the system \eqref{eq:ncP4sym0} can be expressed via solutions of the infinite noncommutative Toda system (see Theorem \ref{thm:ncP4solmainthm} in Section \ref{sec:ncHankel}). In Sections \ref{sec:ncham} and \ref{sec:isompair} we discuss the ``Hamiltonians'' and isomonodromic properties of the system \eqref{eq:ncP4sym0}. We move some proofs to the appendices (see Appendix \ref{sec:app}) to make the main text more straightforward and readable.

\subsection*{Acknowledgments}
The authors are thankful to M.~Mazzocco for sharing her preprint. We are also grateful to
V.~E.~Adler and 
V.~V.~Sokolov for invaluable suggestions and comments. The research of I.B. was partially supported by the International Laboratory of Cluster Geometry HSE, RF Government grant № 075-15-2021-608, and by Young Russian Mathematics award; V.R. is grateful to Centre Henri Lebesgue, program ANR-11-LABX-0020-0; G.Sh. was partially supported by the RNF grant 22-11-00272.

\subsection*{Dedication}
We are very glad to dedicate this paper to our colleague V. V. Sokolov whose profound contributions to the area of non-abelian differential equations are recognized worldwide. Some of his significant works are listed in Appendix \ref{sec:VV}.

\section{Toda equations and their solutions}
\label{sec:Todaeqs}

\subsection{Solutions of the commutative Toda chain} \label{sec:scalToda}
In this section, we review the results from \cite{kajiwara1999determinant}.

The Toda equations can be viewed as the recurrence relation for a sequence of functions $\{\tau_n\}_{n \in \mathbb{Z}}$, $\tau_n = \tau_n (t)$: 
\begin{equation} \label{eq:scalToda}
    \tau_{n + 1} \tau_{n - 1}
    = \tau_n'' \tau_n - {(\tau_n')}^2,
\end{equation}
with general initial conditions $\tau_0$ and $\tau_1$. Equation \eqref{eq:scalToda} is called \textit{the bilinear form of the Toda chain}. It is convenient to introduce a sequence of functions $\{\kappa_n\}_{n \in \mathbb{Z}}$ with $\kappa_0 = 1$ by applying the gauge transformation to $\tau_n$ given by
\begin{equation}
    (\ln \tau_n)''
    = (\ln \kappa_n)'' + \kappa_{-1} \, \kappa_1.
\end{equation}
Then the Toda equations \eqref{eq:scalToda} become
\begin{align} \label{eq:scalToda_kappa}
    &&
    \kappa_{n + 1} \kappa_{n - 1}
    &= \kappa_n'' \kappa_n - {(\kappa_n')}^2 + \kappa_{-1} \kappa_1 \kappa_n^2,
    &
    \kappa_0
    &= 1.
    &&
\end{align}
It tuns out that $\kappa_n$ are uniquely determined for the given $\kappa_{-1}$ and $\kappa_1$; it is also known that $\kappa_n$ are expressible in a determinant form (see Theorem 2.1 in \cite{kajiwara1999determinant}):
\begin{thm} \label{thm:scalkappan}
Let $\{a_n\}_{n \in \mathbb{N}}$, $\{ b_n \}_{n \in \mathbb{N}}$ be two sequences of functions $a_n=a_n(t),\,b_n=b_n(t)$ defined~recursively~as
\begin{align}
    &&
    a_n
    &= a_{n - 1}'
    + \kappa_{-1} \sum_{
    \tiny
    \begin{array}{c}
         i + j = n - 2  
         \\
         i, j \geq 0
    \end{array}
    }
    a_i a_j,
    &
    a_0
    &= \kappa_1,
    &&
    \\[2mm]
    &&
    b_n
    &= b_{n - 1}'
    + \kappa_{1} \sum_{
    \tiny
    \begin{array}{c}
         i + j = n - 2  
         \\
         i, j \geq 0
    \end{array}
    }
    b_i b_j,
    &
    b_0
    &= \kappa_{-1}.
    &&
\end{align}
Let $\kappa_n$ for $n \in \mathbb{Z}$ be an $|n| \times |n|$ Hankel determinant:
\begin{equation}
    \kappa_n
    = \left\{
    \begin{array}{cc}
         \begin{vmatrix}
         a_0 & a_1 & \dots & a_{n - 1} \\
         a_1 & a_2 & \dots & a_n \\
         \vdots & \vdots & \ddots & \vdots \\
         a_{n - 1} & a_n & \dots & a_{2n - 2}
         \end{vmatrix},
         &  
         n > 0,
         \\[2mm]
         1,  &  n = 0 
         \\[2mm]
         \begin{vmatrix}
         b_0 & b_1 & \dots & b_{- n - 1} \\
         b_1 & b_2 & \dots & b_{- n} \\
         \vdots & \vdots & \ddots & \vdots \\
         b_{- n - 1} & b_{- n} & \dots & b_{- 2n - 2}
         \end{vmatrix},
         &  
         n < 0.
    \end{array}
    \right.
\end{equation}
Then $\kappa_n$ satisfies the bilinear Toda equations \eqref{eq:scalToda_kappa} for general initial conditions $\kappa_{-1}$ and $\kappa_1$.
\end{thm}

\begin{proof}
Let us denote the determinant of an $n \times n$ matrix $X$ from the list above by $D$ and let
$D \left(\begin{array}{cccc}
     i_1 & i_2 & \dots & i_k
     \\
     j_1 & j_2 & \dots & j_k
\end{array} \right)
$ be the determinant of the matrix obtained from $X$ by removing $i_1$, $\dots$, $i_k$ rows and $j_1$, $\dots$, $j_k$ columns. Then the statement of Theorem \ref{thm:scalkappan} follows from the relations
\begin{gather}
    \begin{aligned}
    D 
    &\equiv \kappa_{n + 1},
    &&&
    D 
    \begin{pmatrix}
    n + 1  \\ n + 1
    \end{pmatrix}
    &= \kappa_{n},
    &&&
    D 
    \begin{pmatrix}
    n & n + 1 \\ n & n + 1
    \end{pmatrix}
    &= \kappa_{n - 1},
    \end{aligned}
    \\[2mm]
    \begin{aligned}
    D 
    \begin{pmatrix}
    n \\ n + 1
    \end{pmatrix}
    &= D 
    \begin{pmatrix}
    n + 1 \\ n
    \end{pmatrix}
    = \kappa_n',
    &&&&&
    D 
    \begin{pmatrix}
    n \\ n
    \end{pmatrix}
    &= \kappa_n'' + \kappa_{-1} \kappa_1 \kappa_n,
    \end{aligned}
\end{gather}
and the Jacobi identity
\begin{equation}
    D \begin{pmatrix}
    n \\ n
    \end{pmatrix}
    \cdot 
    D \begin{pmatrix}
    n + 1 \\ n + 1
    \end{pmatrix}
    -
    D \begin{pmatrix}
    n \\ n + 1
    \end{pmatrix}
    \cdot 
    D \begin{pmatrix}
    n + 1 \\ n
    \end{pmatrix}
    = D \cdot 
    D \begin{pmatrix}
    n & n + 1 \\ n & n + 1
    \end{pmatrix}.
\end{equation}
\end{proof}

\begin{rem}
Let us put $\theta_n = \kappa_{n} \, \kappa_{n - 1}^{-1}$ for $n \geq 0$, in particular $\theta_0 = \kappa_{-1}^{-1}$ and $\theta_1 = \kappa_1$. Then the Toda equations \eqref{eq:scalToda_kappa} can be rewritten as
\begin{align} \label{eq:scalToda_thetan}
    &&
    \brackets{\ln \theta_n}''
    &= \theta_{n + 1} \, \theta_n^{-1}
    - \theta_n \, \theta_{n - 1}^{-1}
    ,
    &&
\end{align}
Similarly, we can put $\eta_{m} = \kappa_{m} \kappa_{m + 1}^{-1}$ for $ m \leq 0$ to get
\begin{align} \label{eq:scalToda_etam}
    &&
    \brackets{\ln \eta_{m}}''
    &= \eta_{m - 1} \, \eta_{m}^{-1}
    - \eta_{m} \, \eta_{m + 1}^{-1}
    .
    &&
\end{align}
Note that $\eta_0 = \kappa_1^{-1} = \theta_1^{-1}$ and $\eta_{-1} = \kappa_{-1} = \theta_0^{-1}$. In the commutative case, there is no difference between ``positive'' and ``negative'' parts of the Toda equations, but in the noncommutative setting the difference is significant.
\end{rem}

\begin{rem}
The substitution $\theta_n = e^{u_n}$ brings the Toda chain \eqref{eq:scalToda_thetan} to the classical form:
\begin{align}
    u_n''
    &= e^{\text{\small{$u_{n + 1} - u_{n}$}}} 
    - e^{\text{\small{$u_{n} - u_{n - 1}$}}} .
\end{align}
\end{rem}

\subsection{Solutions of noncommutative Toda equations}
\label{sec:ncToda}

In the current section we give an overview of Section 2 from the paper \cite{Retakh_Rubtsov_2010}.

\subsubsection{Quasideterminants and almost Hankel matrices}
Consider a matrix algebra over an associative unital ring $R$ and its element $X = (x_{ij})$, where $1 \leq i, j \leq n$. We denote by $X^{ij}$ the $(n - 1) \times (n - 1)$ matrix obtained from $X$ by removing the $i$-th row and the $j$-th column. Let $r_i$ and $c_j$ be the $i$-th row and the $j$-th column of the matrix $X$:
\begin{align}
    r_i
    &= 
    \begin{pmatrix}
    x_{i1} & x_{i2} & \dots & x_{ij} & \dots & x_{in}
    \end{pmatrix},
    &
    c_j
    &= 
    \begin{pmatrix}
    x_{1 j} & x_{2 j} & \dots & x_{ij} & \dots & x_{n j}
    \end{pmatrix}^T.
\end{align}
According to \cite{gel1991determinants}, \textit{the quasideterminant} $|X|_{ij}$ is defined as
\begin{itemize}
    \item for $n = 1$, $|X|_{11} = x_{11}$;
    \item for $n > 1$ under the condition that $X^{ij}$ is an invertible matrix, $|X|_{ij} = x_{ij} - r_i (X^{ij})^{-1} c_j$.
\end{itemize}
\begin{exmp}
\phantom{}

\begin{itemize}
    \item Consider a $2 \times 2$ matrix $A$ with generic entries $a_{ij}$, $i, j = 1, 2$. Then it has four quasideterminants:
    \begin{align}
        |A|_{11}
        &= a_{11}
        - a_{12} \, a_{22}^{-1} \, a_{21}
        ,
        &
        |A|_{12}
        &= a_{12}
        - a_{11} \, a_{21}^{-1} \, a_{22}
        ,
        \\[1mm]
        |A|_{21}
        &= a_{21}
        - a_{22} \, a_{12}^{-1} \, a_{11}
        ,
        &
        |A|_{22}
        &= a_{22}
        - a_{21} \, a_{11}^{-1} \, a_{12}
        .
    \end{align}
    
    \item In the case of $3 \times 3$ matrix $A$ with generic entries $a_{ij}$, $i, j = 1, 2, 3$, there are nine quasideterminants. For example,
    \begin{align}
        |A|_{11}
        = a_{11}
        &- a_{12} \, \brackets{
        a_{22} 
        - a_{23} \, a_{33}^{-1} \, a_{32}
        }^{-1} a_{21}
        - a_{12} \, \brackets{
        a_{32}
        - a_{33} \, a_{23}^{-1} \, a_{22}
        }^{-1} \, a_{31}
        \\[1mm]
        &- \, a_{13} \brackets{
        a_{23}
        - a_{22} \, a_{32}^{-1} \, a_{33}
        }^{-1} \, a_{21}
        - a_{13} \, \brackets{
        a_{33} 
        - a_{32} \, a_{22}^{-1} \, a_{23}
        }^{-1} \, a_{31}.
    \end{align}
\end{itemize}
\end{exmp}

Note that in the case of commutative ring $R$, the quasideterminant is given simply by
\begin{align}
    &&
    |X|_{ij}
    &= (-1)^{i + j} \frac{\det X}{\det X^{ij}}, 
    &&&
    \text{for any}
    &&
    i, j
    &= 1, 2, \dots, n.
    &&
\end{align}

Following the paper \cite{Retakh_Rubtsov_2010}, we define \textit{almost Hankel matrices} $H_n (i, j) = (x_{st})$, where $s, t = 0, 1, \dots, n$ and $i, j \geq 0$, for a sequence $x_0$, $x_1$, $\dots$ of the elements of $R$ in the following way:
\begin{align}
    x_{st}
    &= x_{s + t},
    &
    x_{nt}
    &= x_{i + t},
    &
    x_{s n}
    &= x_{s + j},
    &
    x_{n n}
    &= x_{i + j},
\end{align}
where $s, t < n$. In particular, $H_n (n,n)$  is a Hankel matrix. 

Let us denote by $h_n (i, j)$ the quasideterminant $|H_n (i, j)|_{nn}$. If at least one of the conditions $i < n$, $j < n$ holds, we have $h_n (i, j) = 0$.

\subsubsection{Noncommutative Toda equations}

Consider an associative unital division ring $R$ over a field $F$ equipped with a derivation $D: R \to R$ such that
\begin{enumerate}
    \item[1)] $D(\alpha) = 0$ for any $\alpha \in F$;
    
    \item[2)] there exists an element $t \in R$ such that $D(t) = 1$.
\end{enumerate}
Let us use the following notation $D(u) = u'$, $D^2(u) = u''$, $\dots$, and recall that $D(v^{-1}) = - v^{-1} v' v^{-1}$ for any invertible element $v \in R$. We assume that all objects related to the noncommutative case belong to the~ring~$R$. Elements of the ring $R$ we will call \textit{functions}.

In the paper \cite{Retakh_Rubtsov_2010} the following noncommutative analogs of the Toda equations \eqref{eq:scalToda_thetan} and \eqref{eq:scalToda_etam} were introduced:
\begin{align}
    \label{eq:ncToda_thetan}
    (\theta_n' \theta_n^{-1})'
    &= \theta_{n + 1} \theta_n^{-1}
    - \theta_n \theta_{n - 1}^{-1},
    &
    n 
    &\geq 0,
    \\[2mm]
    \label{eq:ncToda_etam}
    (\eta_{m}^{-1} \eta_{m}') '
    &= \eta_{m}^{-1} \eta_{m - 1} - \eta_{m + 1}^{-1} \eta_{m},
    & 
    m &\leq 0,
\end{align}
where $\theta_1 = \kappa_{1} = \eta_0^{-1}$ and $\theta_0 = \kappa_{-1}^{-1} = \eta_{-1}^{-1}$ with general initial conditions $\kappa_{-1}$ and $\kappa_1$. Further in the text, we will refer to the expressions $\theta' \theta^{-1}$ and $\theta^{-1} \theta'$ as the left and right noncommutative logarithmic derivatives of $\theta$, respectively.

There is a noncommutative generalization of Theorem \ref{thm:scalkappan}, obtained in \cite{Retakh_Rubtsov_2010}.
\begin{thm} \label{thm:ncthetan}
Let $\{a_n\}_{n \in \mathbb{N}}$, $\{ b_n \}_{n \in \mathbb{N}}$ be two sequences of elements in $R$ defined recursively as
\begin{align}
    \label{eq:ncan}
    &&
    a_n
    &= a_{n - 1}'
    + \sum_{
    \tiny
    \begin{array}{c}
         i + j = n - 2  
         \\
         i, j \geq 0
    \end{array}
    }
    a_i \, \kappa_{-1} \, a_j,
    &
    a_0
    &= \kappa_1,
    &&
    \\[2mm]
    \label{eq:ncbn}
    &&
    b_n
    &= b_{n - 1}'
    + \sum_{
    \tiny
    \begin{array}{c}
         i + j = n - 2  
         \\
         i, j \geq 0
    \end{array}
    }
    b_i \, \kappa_{1} \, b_j,
    &
    b_0
    &= \kappa_{-1}.
    &&
\end{align}
Let $\theta_{p + 1} = |A_p|_{pp}$ and $\eta_{q - 1} = |B_{- q}|_{-q, -q}$, where $A_n = (a_{i + j})$ and $B_n = (b_{i + j})$, $i, j = 0, 1, \dots, n$, are Hankel matrices.

Then $\theta_n$, $n \geq 0$ and $\eta_{m}$, $m \leq 0$ satisfy the systems \eqref{eq:ncToda_thetan} and \eqref{eq:ncToda_etam} respectively for general initial conditions $\kappa_{-1}$ and $\kappa_1$.
\end{thm}

\begin{rem}
The proof of Theorem \ref{thm:ncthetan} is based on the properties of almost Hankel matrices and the noncommutative generalization of the Sylvester identity \cite{gel1991determinants}, \cite{gel1992theory}.
\end{rem}

\begin{exmp}
Let us show that under the conditions of Theorem \ref{thm:ncthetan} the left logarithmic derivative $\theta_1' \, \theta_1^{-1}$ satisfies the noncommutative Toda chain \eqref{eq:ncToda_thetan}:
\begin{align}
    \brackets{\theta_1' \, \theta_{1}^{-1}}'
    &= \theta_2 \, \theta_1^{-1}
    - \theta_1 \, \theta_0^{-1},
\end{align}
where
\begin{align}
    \theta_0
    &= b_0^{-1},
    &
    \theta_1
    &= a_0,
    &
    \theta_2
    &= a_2 
    - a_1 \, a_0^{-1} \, a_1,
\end{align}
and the functions $a_n$ are defined by recurrence relations \eqref{eq:ncan}. Indeed,
\begin{align}
    \brackets{\theta_1' \, \theta_{1}^{-1}}'
    &= \brackets{
    a_0' \, a_0^{-1}
    }'
    = \brackets{a_1 \, a_0^{-1}}'
    = a_1' \, a_0^{-1} 
    - (a_0' \, a_0^{-1})^2
    = \brackets{
    a_2 - a_0 \, b_0 \, a_0
    } a_0^{-1}
    - (a_1 \, a_0^{-1})^2
    \\[1mm]
    &= \brackets{
    a_2 
    - a_1 \, a_0^{-1} \, a_1
    } a_0^{-1}
    - a_0 \, b_0
    = \theta_2 \, \theta_1^{-1}
    - \theta_1 \, \theta_0^{-1}.
\end{align}
Similar computations show that the right logarithmic derivative $\eta_{-1}^{-1} \, \eta_{-1}'$, of Hankel quasideterminant $\eta_{-1}$ (where the entries $b_n$ of $\eta_{-m}$ are defined by \eqref{eq:ncbn}), is a solution of the noncommutative Toda chain \eqref{eq:ncToda_etam}:
\begin{align}
    \brackets{\eta_{-1}^{-1} \, \eta_{-1}'}'
    &= \brackets{b_0^{-1} \, b_0'}'
    = \brackets{b_0^{-1} \, b_1}'
    = - \brackets{b_0^{-1} b_0'}^{-2}
    + b_0^{-1} \, b_1'
    = b_0^{-1} \brackets{
    b_2 - b_0 \, a_0 \, b_0
    }
    - \brackets{b_0^{-1} b_1}^{-2}
    \\[1mm]
    &= b_0^{-1} \brackets{
    b_2 - b_1 \, b_0^{-1} \, b_1
    }
    - a_0 \, b_0
    = \eta_{-1}^{-1} \, \eta_{-2}
    - \eta_0^{-1} \, \eta_{-1}.
\end{align}
\end{exmp}

\section{Commutative \Painleve IV equation}
\label{sec:scalP4}
In this section we consider the \Painleve IV equation written in the form
\begin{align}
    \label{eq:P4y}
    y ''
    &= \tfrac12 y^{-1} (y')^2
    + \tfrac32 y^3 
    - 2 t y^2
    + \brackets{\tfrac12 t^2 - (\alpha_1 - \alpha_0)} y
    - \tfrac12 \alpha_2^2 y^{-1},
\end{align}
where $y = y(t)$ and $\alpha_0 + \alpha_1 + \alpha_2 = 1$, and describe its \Backlund transformations and some of its representations in the Hirota bilinear form (see Sections \ref{sec:scalBT}, \ref{sec:Hirota}). We will follow \cite{okamoto1981tau}, \cite{noumi2004painleve} and refer the reader to these texts for details. 

In Section \ref{sec:scalHankel}, we prove that solutions of the \Painleve IV equation are expressible in terms of solutions of the Toda equations under some assumptions for the general initial functions $\kappa_{-1}$ and $\kappa_1$ (see Theorem \ref{thm:scalP4solmainthm} in Section~\ref{sec:scalHankel}).

\subsection{Symmetric form and \Backlund transformations} 
\label{sec:scalBT}

The symmetric form of the $\PIV$ equation \eqref{eq:P4y} is given by
\begin{align}
    \label{eq:P4sym}
    &\left\{
    \begin{array}{lcl}
         f_0'
         &=& f_0 f_1 - f_0 f_2 + \alpha_0,  
         \\[2mm]
         f_1'
         &=& f_1 f_2 - f_0 f_1 + \alpha_1,
         \\[2mm]
         f_2'
         &=& f_0 f_2 - f_1 f_2 + \alpha_2,
    \end{array}
    \right.
\end{align}
where $f_i = f_i (t)$ and $\alpha_0 + \alpha_1 + \alpha_2 = 1$. This system has the first integral 
$I = f_0 + f_1 + f_2 - t$
and is equivalent to $\PIV$ \eqref{eq:P4y} for $y = f_2$ by eliminating $f_0$ and $f_1$. We define $\tau$-functions $\tau_i$, $i = 0, 1, 2$ by the formulas
\begin{align}
    h_0
    &= \tau_0' \tau_0^{-1},
    &
    h_1
    &= \tau_1' \tau_1^{-1},
    &
    h_2
    &= \tau_2' \tau_2^{-1},
\end{align}
where $h_i$, $i = 0, 1, 2$, are Hamiltonians:
\begin{align}
    \label{eq:scalh0}
    h_0
    &= f_0 f_1 f_2
    + \frac{\alpha_1 - \alpha_2}{3} f_0
    + \frac{\alpha_1 + 2 \alpha_2}{3} f_1
    - \frac{2 \alpha_1 + \alpha_2}{3} f_2,
    \\[1mm]
    h_1
    &= f_0 f_1 f_2
    - \frac{2 \alpha_2 + \alpha_0}{3} f_0
    + \frac{\alpha_2 - \alpha_0}{3} f_1
    + \frac{\alpha_2 + 2 \alpha_0}{3} f_2,
    \\[1mm]
    h_2
    &= f_0 f_1 f_2
    + \frac{\alpha_0 + 2 \alpha_1}{3} f_0
    - \frac{2 \alpha_0 + \alpha_1}{3} f_1
    + \frac{\alpha_0 - \alpha_1}{3} f_2.
\end{align}

The symmetric form \eqref{eq:P4sym} is invariant under the following \Backlund transformations defined by Table \ref{tab:BTonP4sym}:
\begin{table}[H]
    \centering
    \begin{tabular}{c||ccc|ccc}
         & $\alpha_0$
         & $\alpha_1$
         & $\alpha_2$
         & $f_0$
         & $f_1$
         & $f_2$
         \\[1mm]
         \hline\hline
         $s_0$
         & $- \alpha_0$
         & $\alpha_1 + \alpha_0$
         & $\alpha_2 + \alpha_0$
         & $f_0$
         & $f_1 + \tfrac{\alpha_0}{f_0}$
         & $f_2 - \tfrac{\alpha_0}{f_0}$
         \\[2mm]
         $s_1$
         & $\alpha_0 + \alpha_1$
         & $- \alpha_1$
         & $\alpha_2 + \alpha_1$
         & $f_0 - \tfrac{\alpha_1}{f_1}$
         & $f_1$
         & $f_2 + \tfrac{\alpha_1}{f_1}$
         \\[2mm]
         $s_2$
         & $\alpha_0 + \alpha_2$
         & $\alpha_1 + \alpha_2$
         & $- \alpha_2$
         & $f_0 + \tfrac{\alpha_2}{f_2}$
         & $f_1 - \tfrac{\alpha_2}{f_2}$
         & $f_2$
         \\[2mm]
         \hline
         \vspace{4mm}
         $\pi$
         & $\alpha_1$
         & $\alpha_2$
         & $\alpha_0$
         & $f_1$
         & $f_2$
         & $f_0$
    \end{tabular}
    \caption{\Backlund transformations of the symmetric form \eqref{eq:P4sym}}
    \label{tab:BTonP4sym}
\end{table}

\Backlund transformations form an extended affine Weyl group of type $A_2^{(1)}$:
\begin{equation}
    \tilde{W} \brackets{A_2^{(1)}} 
    = \angleBrackets{s_0, s_1, s_2; \pi},
\end{equation}
where the generators $s_i$, $i = 0, 1, 2$ and $\pi$ commute with the derivation and satisfy the following fundamental relations
\begin{align}
    &&
    s_i^2
    &= 1,
    &
    (s_i s_{i + 1})^3
    &= 1,
    &
    \pi^3
    &= 1,
    &
    \pi s_i
    &= s_{i + 1} \pi,
    &
    i
    &\in \mathbb{Z}_3.
    &&
\end{align}

Using the relations between the $f$-variables and the $\tau$-functions
\begin{align}
    f_0
    &= (\ln \tau_2 \tau_1^{-1})' + \tfrac13 t,
    &
    f_1
    &= (\ln \tau_0 \tau_2^{-1})' + \tfrac13 t,
    &
    f_2
    &= (\ln \tau_1 \tau_0^{-1})' + \tfrac13 t,
\end{align}
one can show that $\tau$-functions satisfy different bilinear differential equations of the Hirota type. Here are some of them:
\begin{align} \label{eq:biltau1}
\begin{aligned}
    \brackets{
    D_t^2 + \tfrac13 t D_t - \tfrac29 t^2 + \tfrac13 (\alpha_0 - \alpha_1)
    } \tau_0 \cdot \tau_1
    &= 0,
    \\[1mm]
    \brackets{
    D_t^2 + \tfrac13 t D_t - \tfrac29 t^2 + \tfrac13 (\alpha_1 - \alpha_2)
    } \tau_1 \cdot \tau_2
    &= 0,
    \\[1mm]
    \brackets{
    D_t^2 + \tfrac13 t D_t - \tfrac29 t^2 + \tfrac13 (\alpha_2 - \alpha_0)
    } \tau_2 \cdot \tau_0
    &= 0,
\end{aligned}
\end{align}
and
\begin{align}
\begin{aligned} \label{eq:biltau2}
    \brackets{
    \tfrac12 D_t^2 - \tfrac13 (\alpha_1 - \alpha_2)
    } \tau_0 \cdot \tau_0
    &= s_1 (\tau_1) \, s_2 (\tau_2),
    \\[1mm]
    \brackets{
    \tfrac12 D_t^2 - \tfrac13 (\alpha_1 - \alpha_2)
    } \tau_1 \cdot \tau_1
    &= s_2 (\tau_2) \, s_0 (\tau_0),
    \\[1mm]
    \brackets{
    \tfrac12 D_t^2 - \tfrac13 (\alpha_1 - \alpha_2)
    } \tau_2 \cdot \tau_2
    &= s_0 (\tau_0) \, s_1 (\tau_1),
\end{aligned}
\end{align}
where $D_t^n$ is the Hirota operator:
\begin{align}
    D_t^n \, f \cdot g
    &= \left.
    \brackets{
    \frac{d}{dt} - \frac{d}{ds}
    }^n f(t) \, g(s) \right|_{s = t}.
\end{align}

\subsection{Translations and \texorpdfstring{$\tau$}{tau}-functions on the lattice}
\label{sec:Hirota}

By composing \Backlund transformations, we define the translation operators $T_i$, $i = 1, 2, 3$:
\begin{align}
    T_1
    &= \pi s_2 s_1,
    &
    T_2
    &= s_1 \pi s_2,
    &
    T_3
    &= s_2 s_1 \pi.
\end{align}
From this definition, it follows that $T_1 T_2 T_3 = 1$ and
\begin{align}
    T_1^l (\alpha_0)
    &= \alpha_0 + l,
    &
    T_1^l (\alpha_1)
    &= \alpha_1 - l,
    &
    T_1^l (\alpha_2)
    &= \alpha_2;
    \\[1mm]
    T_2^m (\alpha_0)
    &= \alpha_0,
    &
    T_2^m (\alpha_1)
    &= \alpha_1 + m,
    &
    T_2^m (\alpha_2)
    &= \alpha_2 - m;
    \\[1mm]
    T_3^n (\alpha_0)
    &= \alpha_0 - n,
    &
    T_3^n (\alpha_1)
    &= \alpha_1,
    &
    T_3^n (\alpha_2)
    &= \alpha_2 + n.
\end{align}
Thus, we define the $\tau$-functions on the lattice by
\begin{align}
    \tau_{l, m, n}
    &= T_1^l T_2^m T_3^n (\tau_0).
\end{align}
In particular, 
\begin{align}
    \tau_{0, 0, 0}
    &= \tau_0,
    &
    \tau_{1, 0, 0}
    &= \tau_1,
    &
    \tau_{1, 1, 0}
    &= \tau_2.
\end{align}
Since $T_1T_2 T_3 = 1$, for any $k \in \mathbb{Z}$ we have $\tau_{l + k, m + k, n + k} = \tau_{l, m, n}$. Applying $T_1^l T_2^m T_3^n$ on the bilinear equations \eqref{eq:biltau1} -- \eqref{eq:biltau2}, we obtain the following equations:
\begin{align} \label{eq:biltaun1}
\begin{aligned}
    &\brackets{
    D_t^2 + \tfrac13 t D_t - \tfrac29 t^2 + \tfrac13 (\alpha_0 - \alpha_1 + 2 l - m - n)
    } \tau_{l, m, n} \cdot \tau_{l + 1, m, n}
    = 0,
    \\[1mm]
    &\brackets{
    D_t^2 + \tfrac13 t D_t - \tfrac29 t^2 + \tfrac13 (\alpha_1 - \alpha_2 - l + 2 m - n)
    } \tau_{l + 1, m, n} \cdot \tau_{l + 1, m + 1, n}
    = 0,
    \\[1mm]
    &\brackets{
    D_t^2 + \tfrac13 t D_t - \tfrac29 t^2 + \tfrac13 (\alpha_2 - \alpha_0 - l - m + 2 n)
    } \tau_{l + 1, m + 1, n} \cdot \tau_{l, m, n}
    = 0,
\end{aligned}
\end{align}
and
\begin{align}
\begin{aligned} \label{eq:biltaun2}
    &\brackets{
    \tfrac12 D_t^2 - \tfrac13 (\alpha_1 - \alpha_2 - l + 2 m - n)
    } \tau_{l, m, n} \cdot \tau_{l, m, n}
    = \tau_{l, m + 1, n} \, \tau_{l, m - 1, n}
    \\[1mm]
    &\brackets{
    \tfrac12 D_t^2 - \tfrac13 (\alpha_1 - \alpha_2 - l - m + 2 n)
    } \tau_{l + 1, m, n} \cdot \tau_{l + 1, m, n}
    = \tau_{l + 1, m, n + 1} \, \tau_{l + 1, m , n - 1}
    \\[1mm]
    &\brackets{
    \tfrac12 D_t^2 - \tfrac13 (\alpha_1 - \alpha_2 + 2 l - m - n)
    } \tau_{l + 1, m + 1, n} \cdot \tau_{l + 1, m + 1, n}
    = \tau_{l + 2, m + 1, n} \, \tau_{l, m + 1, n}.
\end{aligned}
\end{align}

Following the notation in Section \ref{sec:scalToda}, we introduce the variables $\kappa_{l, m, n}$ by
\begin{align}
    \kappa_{l, m, n}
    &= e^{- \frac13 (l + m + n) t^2} \tau_{l, m, n} \, \tau_{0, 0, 0}^{-1}.
\end{align}
Here $\kappa_{l, m, n}$ depend on three indices corresponding to the three operators $T_1$, $T_2$, and $T_3$ (unlike in the paper \cite{joshi2006generating}): shift of the $i$-th index corresponds to the action of the operator $T_i,\,i=1,2,3$. Then the bilinear equations \eqref{eq:biltaun1} -- \eqref{eq:biltaun2} become
\begin{align}
    \label{eq:bilkap0kap1}
    &\brackets{
    D_t^2 
    - t D_t 
    + 2 \kappa_{-1, 0, 0} \, \kappa_{1, 0, 0}
    + \brackets{
    \alpha_0 
    - \alpha_1 
    + 2 l 
    + m 
    + n
    }
    } 
    \kappa_{l, m, n} 
    \cdot 
    \kappa_{l + 1, m, n}
    = 0,
    \\[1mm]
    \label{eq:bilkap1kap2}
    &\brackets{
    D_t^2 
    - t D_t 
    + 2 \kappa_{0, -1, 0} \, \kappa_{0, 1, 0}
    + \brackets{
    \alpha_1 - \alpha_2
    + l 
    + 2 m 
    + n
    + 2
    }
    } 
    \kappa_{l + 1, m, n } 
    \cdot 
    \kappa_{l + 1, m + 1, n}
    = 0,
    \\[1mm]
    \label{eq:bilkap2kap0}
    &
    \brackets{
    D_t^2 
    - t D_t 
    + 2 \kappa_{0, 0, -1} \, \kappa_{0, 0, 1}
    + \brackets{
    \alpha_2
    - \alpha_0 
    + l 
    + m 
    + 2 n
    + 4
    }
    } 
    \kappa_{l + 1, m + 1, n} 
    \cdot 
    \kappa_{l + 1, m + 1, n + 1}
    = 0,
\end{align}
\begin{align}
    \label{eq:kapToda1}
    &\brackets{
    \tfrac12 D_t^2
    + \kappa_{-1, 0, 0} \, \kappa_{1, 0, 0}
    + (m + n)
    } 
    \kappa_{l, m, n} \cdot
    \kappa_{l, m, n}
    = \kappa_{l + 1, m, n} 
    \,
    \kappa_{l - 1, m, n},
    \\[1mm]
    \label{eq:kapToda2}
    &\brackets{
    \tfrac12 D_t^2
    + \kappa_{0, -1, 0} \, \kappa_{0, 1, 0}
    + (l + n)
    } 
    \kappa_{l, m, n} \cdot
    \kappa_{l, m, n}
    = \kappa_{l, m + 1, n} 
    \,
    \kappa_{l, m - 1, n},
    \\[1mm] 
    \label{eq:kapToda3}
    &\brackets{
    \tfrac12 D_t^2
    + \kappa_{0, 0, -1} \, \kappa_{0, 0, 1}
    + (l + m)
    } 
    \kappa_{l, m, n} \cdot
    \kappa_{l, m, n}
    = \kappa_{l, m, n + 1} 
    \,
    \kappa_{l, m, n - 1}.
\end{align}

\subsection{Hankel determinant formula}
\label{sec:scalHankel}

In this section we consider a sequence of the $\kappa$-functions $\kappa_{n, 0, 0} = \kappa_n$, $n \in \mathbb{Z}$, only in the $T_1$-direction
(similar results can be formulated in the directions of the operators $T_2$ and $T_3$). This sequence can be regarded as being generated by the Toda equations \eqref{eq:kapToda1}:
\begin{align} \label{eq:kapTodan}
    &&
    \brackets{
    \tfrac12 D_t^2
    + \kappa_{-1} \, \kappa_1
    } \kappa_{n} \cdot \kappa_n
    &= \kappa_{n - 1} \, \kappa_{n + 1},
    &
    \kappa_0
    &= 1,
    &&
\end{align}
where $\kappa_{-1}$ and $\kappa_1$ are general initial conditions. As we have mentioned in Section \ref{sec:scalToda}, solutions of this equation are given by the Hankel determinant (see Theorem \ref{thm:scalkappan}). For the sake of completeness, we are going to show that solutions of the $\PIV[y_{n}; n]$ equation
\begin{align} \label{eq:P4yn}
    y_n ''
    &= \tfrac12 y_n^{-1} (y_n')^2
    + \tfrac32 y_n^3 
    - 2 t y_n^2
    + \brackets{\tfrac12 t^2 + \alpha_0 - \alpha_1 + 2 n} y_n
    - \tfrac12 \alpha_2^2 y_n^{-1},
\end{align}
where 
$$y_n = (\ln \kappa_{n + 1} \, \kappa_n^{-1})' + t,$$
are also given by the Hankel determinants, if we impose some conditions for $\kappa_{-1}$ and $\kappa_1$.

\begin{prop}
\label{thm:ym1y0sol}
Let $\kappa_{-1}$ and $\kappa_1$ satisfy the following relations
\begin{align}
    \label{eq:kapm1cond}
    &\kappa_{-1}'' - t \kappa_{-1}' + 2 \kappa_{-1}^2 \kappa_{1} + \brackets{\alpha_0 - \alpha_1 - 2} \kappa_{-1}
    = 0,
    \\[1mm]
    \label{eq:kap1cond}
    &\kappa_{1}'' + t \kappa_{1}' + 2 \kappa_{-1} \kappa_{1}^2 + \brackets{\alpha_0 - \alpha_1} \kappa_1
    = 0.
\end{align}  
Let $z_0 = \kappa_{-1} \kappa_{1} - (\alpha_1 + \alpha_2)$. 
Then
\begin{itemize}
    \item[(a)] $y_{0} = \kappa_{1}' \kappa_{1}^{-1} + t$ is a solution of $\PIV[y_0; 0]$ \eqref{eq:P4yn}, if 
    \begin{align} \label{eq:zy0cond}
        - z_0'
        &= y_{0}^{-1} z_0^2 
        + (\alpha_2 - y_{0}^2) y_{0}^{-1} z_0 - (\alpha_1 + \alpha_2) y_{0};
    \end{align}
    
    \item[(b)] $y_{-1} = - \kappa_{-1}' \kappa_{-1}^{-1} + t$ is a solution of $\PIV[y_{-1}; -1]$ \eqref{eq:P4yn}, if 
    \begin{align} \label{eq:zym1cond}
        z_0'
        &= y_{-1}^{-1} z_0^2 
        + (\alpha_2 - y_{-1}^2) y_{-1}^{-1} z_0 - (\alpha_1 + \alpha_2) y_{-1}.
    \end{align}
\end{itemize}
\end{prop}
\begin{proof}
First of all, we notice that the conditions for $\kappa_{-1}$ and $\kappa_1$ follow from the bilinear equation \eqref{eq:bilkap0kap1} for $(l, m, n) = (-1, 0, 0)$ and $(l, m, n) = (0, 0, 0)$.

\medskip
\textbullet \,\, \textbf{Case (a).} 
Consider the case of $y_{0}$. Using the definition of $y_0$, we take the derivative w.r.t. $t$:
\begin{align*}
    y_{0}'
    &= (\kappa_1' \, \kappa_1^{-1} + t)' 
    = \kappa_1'' \kappa_1^{-1} - (\kappa_1' \kappa_1^{-1})^2 + 1.
\end{align*}
Then we use the condition for $\kappa_1$ to reduce the order of the relation:
\begin{align*}
    y_{0}'
    &= \kappa_1'' \kappa_1^{-1} 
    - (\kappa_1' \kappa_1^{-1})^2 
    + 1
    = - t \kappa_1' \kappa_1^{-1} - 2 \kappa_{-1} \kappa_1  
    - (\kappa_{1}' \kappa_1^{-1})^2 
    + (2 \alpha_1 + \alpha_2).
\end{align*}
Substituting $\kappa_{1}' \, \kappa_{1}^{-1} = y_0 - t$ and $\kappa_{-1} \, \kappa_1 = z_0 + (\alpha_1 + \alpha_2)$, we arrive at the equation
\begin{align*}
    y_{0}'
    &= - 2 z_0 - y_0^2 + t y_0 - \alpha_2,
\end{align*}
that, with the assumption \eqref{eq:zy0cond} for $z_0$, gives the following system
\begin{align} \label{eq:y0sys}
    &\left\{
    \begin{array}{lcl}
         - z_0'
         &=& y_{0}^{-1} z_0^2 + (\alpha_2 - y_{0}^2) y_{0}^{-1} z_0 - (\alpha_1 + \alpha_2) y_{0},
         \\[2mm]
         - y_0'
         &=& 2 z_0 + y_0^2 - t y_0 + \alpha_2,
    \end{array}
    \right.
\end{align}
which is equivalent to the $\PIV[y_0; 0]$ equation \eqref{eq:P4yn} with $n = 0$:
\begin{equation}
    y_0''
    = \tfrac12 y_0^{-1} (y_0')^2 + \tfrac32 y_0^3 - 2 t y_0^2 + \brackets{
    \tfrac12 t^2 + \alpha_0 - \alpha_1
    } y_0 - \tfrac12 \alpha_2^2 y_0^{-1}.
\end{equation}

\medskip
\textbullet \,\, \textbf{Case (b).}
The case of $y_{-1}$ can be regarded in a similar way. The resulting system is given by
\begin{align} \label{eq:ym1sys}
    &\left\{
    \begin{array}{lcl}
         z_0'
         &=& y_{-1}^{-1} z_0^2 + (\alpha_2 - y_{-1}^2) y_{-1}^{-1} z_0 - (\alpha_1 + \alpha_2) y_{-1},
         \\[2mm]
         y_{-1}'
         &=& 2 z_0 + y_{-1}^2 - t y_{-1} + \alpha_2,
    \end{array}
    \right.
\end{align}
where the first equation is just condition \eqref{eq:zym1cond}.
Eliminating $z$ from this system, we arrive at the~$\PIV[y_{-1}; -1]$ equation \eqref{eq:P4yn} with $n = -1$:
\begin{equation}
    y_{-1}''
    = \tfrac12 y_{-1}^{-1} (y_{-1}')^2 + \tfrac32 y_{-1}^3 - 2 t y_{-1}^2 + \brackets{
    \tfrac12 t^2 + \alpha_0 - \alpha_1 - 2
    } y_{-1} - \tfrac12 \alpha_2^2 y_{-1}^{-1}.
\end{equation}
\end{proof}

\begin{rem}
If we substitute $y_{-1}$ and $z_0$ into \eqref{eq:zym1cond} or $y_{0}$ and $z_0$ into \eqref{eq:zy0cond}, we get the following additional condition for $\kappa_{-1}$, $\kappa_1$:
    \begin{align} \label{eq:kapm1kap1thirdcond}
        \kappa_{-1}' \, \kappa_1'
        + t \brackets{
        \kappa_{-1}' \, \kappa_1
        - \kappa_{-1} \, \kappa_1'
        }
        + \kappa_{-1}^2 \kappa_1^2
        - (t^2 - \alpha_0 + \alpha_1 + 1) \kappa_{-1} \, \kappa_1
        &= \alpha_1 (\alpha_0 - 1).
    \end{align}
So, as in the case of the $\text{PII}$ equation (see Proposition 2.2 in \cite{joshi2004generating}), the initial functions $\kappa_{-1}$ and $\kappa_1$ should satisfy three conditions: \eqref{eq:kapm1cond}, \eqref{eq:kap1cond}, and \eqref{eq:kapm1kap1thirdcond}.

Condition \eqref{eq:kapm1kap1thirdcond} can be derived without the use of the auxiliary function $z_0$. If we require that $y_0 = \kappa_1' \, \kappa_1^{-1} + t$ (resp. $y_{-1} = - \kappa_{-1}' \, \kappa_{-1}^{-1} + t$) is a solution of the $\PIV[y_0; 0]$ (resp. $\PIV[y_{-1}; - 1]$) equation, where $\kappa_{- 1}$ and $\kappa_{1}$ satisfy \eqref{eq:kapm1cond} -- \eqref{eq:kap1cond}, then $\kappa_{-1}$, $\kappa_1$ must satisfy condition \eqref{eq:kapm1kap1thirdcond}. The converse statement is also true, but it is more convenient to prove it in other notation and using \Backlund transformations.
\end{rem}

A generalization of Proposition \ref{thm:ym1y0sol} to the case of arbitrary $n \in \mathbb{Z}$ is given in Theorem \ref{thm:scalkapnsol}:
\begin{thm} \label{thm:scalkapnsol}
Let $\kappa_n$, $n \in \mathbb{Z}$, satisfy the Toda equations \eqref{eq:kapTodan} and the bilinear equation
\begin{align}
    \label{eq:bilkapncond}
    \brackets{
    D_t^2 - t D_t + 2 \kappa_{-1} \, \kappa_1 + (\alpha_0 - \alpha_1 + 2 n)
    } \kappa_{n} \cdot \kappa_{n + 1}
    &= 0.
\end{align}
Let $y_n = (\ln \kappa_{n + 1} \, \kappa_{n}^{-1})' + t$ and the function $z_n$ be defined as 
\begin{equation} \label{eq:zn}
    z_n 
    = \kappa_{n - 1} \, \kappa_n^{-2} \, \kappa_{n + 1} - (\alpha_1 + \alpha_2 - n).
\end{equation}
Then 
\begin{itemize}
    \item[(a)] for $n \geq 0$, $y_n = y_n (t)$ is a solution of the $\PIV[y_n; n]$ equation \eqref{eq:P4yn}, if 
    \begin{align} \label{eq:znposcond}
        - z_n' 
        &= y_n^{-1} z_n^2
        + \brackets{\alpha_2 - y_n^2} y_n^{-1} z_n
        - \brackets{\alpha_1 + \alpha_2 - n} y_n;
    \end{align}
    
    \item[(b)] for $n \leq 0$, $y_{n - 1} = y_{n - 1} (t)$ is a solution of the $\PIV[y_{n - 1}; n - 1]$ equation \eqref{eq:P4yn}, if 
    \begin{align} \label{eq:znnegcond}
        z_n' 
        &= y_{n - 1}^{-1} z_n^2
        + \brackets{\alpha_2 - y_{n - 1}^2} y_{n - 1}^{-1} z_n
        - \brackets{\alpha_1 + \alpha_2 - n} y_{n - 1}.
    \end{align}
\end{itemize}
\end{thm}
\begin{proof}
The proof is given by a straightforward computation, in which we will use the bilinear equation \eqref{eq:bilkapncond}:
\begin{align}
    \kappa_n'' \, \kappa_{n}^{-1} 
    - 2 \kappa_n' \, \kappa_n^{-1} \, \kappa_{n + 1}' \, \kappa_{n + 1}^{-1}
    + \kappa_{n + 1}'' \, \kappa_{n + 1}^{-1}
    - t \brackets{
    \kappa_n' \, \kappa_{n}^{-1}
    - \kappa_{n + 1}' \, \kappa_{n + 1}^{-1}
    }
    + 2 \kappa_{-1} \, \kappa_1
    + (\alpha_0 - \alpha_1 + 2 n)
    &= 0,
\end{align}
which is just the equation \eqref{eq:bilkap0kap1} for the set $(n, 0, 0)$, and the Toda equations \eqref{eq:kapTodan}:
\begin{align}
    \kappa_n'' \, \kappa_n^{-1} 
    - {(\kappa_n' \, \kappa_n^{-1})}^2
    + \kappa_{-1} \, \kappa_1
    &= \kappa_{n - 1} \, \kappa_n^{-2} \, \kappa_{n + 1}.
\end{align}
Note that case (a) and case (b) coincide with statements in cases (a) and (b) from Proposition \ref{thm:ym1y0sol} for $n = 0$.

\medskip
\textbullet \,\, \textbf{Case (a).} 
Take the derivative of $y_n = \kappa_{n + 1}' \, \kappa_{n + 1}^{-1} - \kappa_{n}' \, \kappa_n^{-1} + t$ w.r.t. $t$:
\begin{align}
    y_n'
    &= \brackets{
    \kappa_{n + 1}' \, \kappa_{n + 1}^{-1} 
    - \kappa_{n}' \, \kappa_n^{-1} 
    + t
    }'
    = \kappa_{n + 1}'' \, \kappa_{n + 1}^{-1}
    - {(\kappa_{n + 1}' \, \kappa_{n + 1}^{-1})}^2 
    - \kappa_{n}'' \, \kappa_{n}^{-1}
    + {(\kappa_{n}' \, \kappa_{n}^{-1})}^2
    + 1.
\end{align}
From the bilinear equation it follows that
\begin{align}
    \kappa_{n + 1}'' \, \kappa_{n + 1}^{-1}
    &= - \kappa_n'' \, \kappa_{n}^{-1} 
    + 2 \kappa_n' \, \kappa_n^{-1} \, \kappa_{n + 1}' \, \kappa_{n + 1}^{-1}
    + t \brackets{
    \kappa_n' \, \kappa_{n}^{-1}
    - \kappa_{n + 1}' \, \kappa_{n + 1}^{-1}
    }
    - 2 \kappa_{-1} \, \kappa_1
    - (\alpha_0 - \alpha_1 + 2 n).
\end{align}
Thus,
\begin{align}
    y_n'
    = - 2 \kappa_{n}'' \, \kappa_{n}^{-1} 
    - 2 \kappa_{-1} \, \kappa_1
    + 2 \kappa_n' \, \kappa_n^{-1} \, \kappa_{n + 1}' \, \kappa_{n + 1}^{-1}
    + t \brackets{
    \kappa_n' \, \kappa_{n}^{-1}
    - \kappa_{n + 1}' \, \kappa_{n + 1}^{-1}
    }
    - {(\kappa_{n + 1}' \, \kappa_{n + 1}^{-1})}^2 
    + {(\kappa_{n}' \, \kappa_{n}^{-1})}^2
    \\[1mm]
    + \brackets{
    2 \alpha_1 + \alpha_2 - 2 n
    },
\end{align}
where we have used the condition $\alpha_0 + \alpha_1 + \alpha_2 = 1$. Since $\kappa_n$ satisfy the Toda equations, we replace $\kappa_{n}'' \, \kappa_{n}^{-1} + \kappa_{-1} \, \kappa_1$ by $\kappa_{n - 1} \, \kappa_n^{-2} \, \kappa_{n + 1} + {(\kappa_n' \, \kappa_n^{-1})}^2$. The result can be written as
\begin{align}
    y_n'
    &= - 2 \kappa_{n - 1} \, \kappa_n^{-2} \, \kappa_{n + 1} 
    - t \brackets{
    \kappa_{n + 1}' \, \kappa_{n + 1}^{-1}
    - \kappa_n' \, \kappa_{n}^{-1}
    }
    - \brackets{
     \kappa_{n + 1}' \, \kappa_{n + 1}^{-1}
    - \kappa_n' \, \kappa_n^{-1}
    }^2
    + \brackets{
    2 \alpha_1 + \alpha_2 - 2 n
    }.
\end{align}
Recall that from the definitions of $y_n$ and $z_n$, it follows that
$$(\ln \kappa_{n + 1} \, \kappa_n^{-1})' = y_n - t$$ 
and $\kappa_{n - 1} \, \kappa_n^{-2} \, \kappa_{n + 1}  = z_n + (\alpha_1 + \alpha_2 - n)$. Therefore, we obtain
\begin{align}
    y_n'
    &= - 2 \brackets{
    z_n + (\alpha_1 + \alpha_2 - n)
    }
    - t (y_n - t)
    - (y_n - t)^2
    + \brackets{
    2 \alpha_1 + \alpha_2 - 2 n
    }
    = - 2 z_n - y_n^2 + t y_n - \alpha_2.
\end{align}
This equation with condition \eqref{eq:znposcond} give the system
\begin{align}
    \left\{
    \begin{array}{lcl}
         - z_n'
         &=& y_{n}^{-1} z_n^2 
         + (\alpha_2 - y_{n}^2) y_{n}^{-1} z_n 
         - (\alpha_1 + \alpha_2 - n) y_{n},
         \\[2mm]
         - y_{n}'
         &=& 2 z_n + y_{n}^2 - t y_{n} + \alpha_2,
    \end{array}
    \right.
\end{align}
that is equivalent to the $\PIV[y_n;n]$ equation \eqref{eq:P4yn}:
\begin{equation}
    y_n''
    = \tfrac12 y_n^{-1} {(y_n')}^2
    + \tfrac32 y_n^3 
    - 2 t y_n^2
    + \brackets{
    \tfrac12 t^2 + \alpha_0 - \alpha_1 + 2 n
    } y_n
    - \tfrac12 \alpha_2^2 y_n^{-1}.
\end{equation}

\medskip
\textbullet \,\, \textbf{Case (b).}
Similarly, for the case (b), we have the following chain of identities:
\begin{align}
    y_{n - 1}'
    &
    \begin{aligned}
    = \kappa_n'' \, \kappa_n^{-1} 
    - {(\kappa_{n}' \, \kappa_n^{-1})}^2
    - \kappa_{n - 1}'' \, \kappa_{n - 1}^{-1} 
    + {(\kappa_{n - 1}' \, \kappa_{n - 1}^{-1})}^2
    + 1
    \end{aligned}
    \\[2mm]
    &
    \begin{aligned}
    = 2 \kappa_n'' \, \kappa_n^{-1} + 2 \kappa_{-1} \, \kappa_1
    - {(\kappa_n' \, \kappa_n^{-1})}^2
    + {(\kappa_{n - 1}' \, \kappa_{n - 1}^{-1})}^2
    - 2 \kappa_{n - 1}' \, \kappa_{n - 1}^{-1} \, \kappa_n' \, \kappa_n^{-1}
    \\[1mm]
    - t \brackets{
    \kappa_{n - 1}' \, \kappa_{n - 1}^{-1}
    - \kappa_n' \, \kappa_{n}^{-1}
    }
    + (- 2 \alpha_1 - \alpha_2 + 2 n)
    \end{aligned}
    \\[2mm]
    &
    \begin{aligned}
    = 2 \kappa_{n - 1} \, \kappa_n^{-2} \, \kappa_{n + 1}
    + \brackets{
    \kappa_n' \, \kappa_n^{-1}
    - \kappa_{n - 1}' \, \kappa_{n - 1}^{-1}
    }^2
    + t \brackets{
    \kappa_n' \, \kappa_n^{-1}
    - \kappa_{n - 1}' \, \kappa_{n - 1}^{-1}
    }
    \\[1mm]
    + (- 2 \alpha_1 - \alpha_2 + 2 n).
    \end{aligned}
\end{align}
Replacing $(\ln \kappa_n \, \kappa_{n - 1}^{-1})'$ by $y_{n - 1} - t$ and $\kappa_{n - 1} \, \kappa_n^{-2} \, \kappa_{n + 1}$ by $z_n + (\alpha_1 + \alpha_2 - n)$ and taking into account the condition \eqref{eq:znnegcond}, one can obtain the system
\begin{align}
    \left\{
    \begin{array}{lcl}
         z_n'
         &=& y_{n - 1}^{-1} z_n^2 + (\alpha_2 - y_{n - 1}^2) y_{n - 1}^{-1} z_n - (\alpha_1 + \alpha_2 - n) y_{n - 1},
         \\[2mm]
         y_{n - 1}'
         &=& 2 z_n + y_{n - 1}^2 - t y_{n - 1} + \alpha_2,
    \end{array}
    \right.
\end{align}
which is also equivalent to the $\PIV[y_{n - 1}; n - 1]$ equation \eqref{eq:P4yn}:
\begin{equation}
    y_{n - 1}''
    = \tfrac12 y_{n - 1}^{-1} {(y_{n - 1}')}^2
    + \tfrac32 y_{n - 1}^3 
    - 2 t y_{n - 1}^2
    + \brackets{
    \tfrac12 t^2 + \alpha_0 - \alpha_1 + 2 (n - 1)
    } y_{n - 1}
    - \tfrac12 \alpha_2^2 y_{n - 1}^{-1}.
\end{equation}
\end{proof}

For our next goal, it is convinient to introduce variables $\theta_{n} = \kappa_{n} \, \kappa_{n - 1}^{-1}$, $n \geq 0$, and $\eta_{m} = \kappa_{m} \, \kappa_{m + 1}^{-1}$, $m \leq 0$. Then bilinear equation \eqref{eq:bilkapncond} can be rewritten as
\begin{align}
    &\theta_{n + 1}''
    + t \, \theta_{n + 1}'
    + 2 \theta_{n + 1}^2 \, \theta_n^{-1}
    + \brackets{\alpha_0 - \alpha_1 + 2 n} \theta_{n + 1} 
    = 0,
    \\[2mm]
    &\eta_{m - 1}''
    - t \, \eta_{m - 1}'
    + 2 \eta_{m - 1}^2 \, \eta_{m}^{-1}
    + (\alpha_0 - \alpha_1 + 2 (m - 1)) \eta_{m - 1}
    = 0.
\end{align}
For instance, in the case of $\theta_{n + 1}$, we have the following identities:
\begin{align}
    \kappa_n^{-2}
    \brackets{
    D_t^2 - t D_t + 2 \kappa_{-1} \, \kappa_1 + (\alpha_0 - \alpha_1 + 2 n)
    } \kappa_{n} \cdot \kappa_{n + 1}
    &= 0,
    \\[1mm]
    \theta_{n + 1}''
    + 2 \theta_{n + 1} \kappa_{n}^{-2} 
    \brackets{
    \kappa_n'' \kappa_n - (\kappa_n')^2
    }
    + t \theta_{n + 1}'
    + 2 \kappa_{-1} \kappa_1 \theta_{n + 1}
    + (\alpha_0 - \alpha_1 + 2 n) \theta_{n + 1}
    &= 0,
    \\[1mm]
    \theta_{n + 1}''
    + 2 \theta_{n + 1} \kappa_{n}^{-2} 
    \brackets{
    \kappa_{n - 1} \kappa_{n + 1} - \kappa_{-1} \kappa_1 \kappa_n^2
    }
    + t \theta_{n + 1}'
    + 2 \kappa_{-1} \kappa_1 \theta_{n + 1}
    + (\alpha_0 - \alpha_1 + 2 n) \theta_{n + 1}
    &= 0,
    \\[1mm]
    \theta_{n + 1}''
    + 2 \theta_{n + 1}^2 \theta_n^{-1}
    + t \theta_{n + 1}'
    + (\alpha_0 - \alpha_1 + 2 n) \theta_{n + 1}
    &= 0.
\end{align}
where we used the Toda chain \eqref{eq:kapTodan} and the definition of $\theta_{n + 1}$. 

\begin{rem}
The resulting conditions, \eqref{eq:thetancond} and \eqref{eq:etamcond}, are consequences of the bilinear equation \eqref{eq:bilkapncond} and the Toda equations~\eqref{eq:kapTodan}, but they implicitly require the Toda chain to hold.
\end{rem}

Theorem \ref{thm:scalkapnsol} can be reformulated in the following way.

\begin{thm} \label{thm:scalthetanetamsol}
Let $\kappa_n$, $n \in \mathbb{Z}$, be solution of the Toda chain \eqref{eq:kapTodan}. Assume that the functions 
\begin{align}
    &&
    \theta_{n} 
    &= \kappa_{n} \, \kappa_{n - 1}^{-1}, 
    &
    n 
    &\geq 0,
    &&&
    &\text{and}
    &&&
    \eta_{m} 
    &= \kappa_{m} \, \kappa_{m + 1}^{-1},
    &
    m 
    &\leq 0,
    &&
\end{align} 
satisfy the following equations
\begin{align}
    \label{eq:thetancond}
    &\theta_{n + 1}''
    + t \, \theta_{n + 1}'
    + 2 \theta_{n + 1}^2 \, \theta_n^{-1}
    + \brackets{\alpha_0 - \alpha_1 + 2 n} \theta_{n + 1} 
    = 0,
    \\[2mm]
    \label{eq:etamcond}
    &\eta_{m - 1}''
    - t \, \eta_{m - 1}'
    + 2 \eta_{m - 1}^2 \, \eta_{m}^{-1}
    + (\alpha_0 - \alpha_1 + 2 (m - 1)) \eta_{m - 1}
    = 0.
\end{align}
Under these conditions
\begin{itemize}
    \item[(a)] 
    if the function $$z_n = \theta_n^{-1} \, \theta_{n + 1} - (\alpha_1 + \alpha_2 - n)$$ 
    satisfies the equation
    \begin{align}
    \label{eq:zncond}
        - z_n' 
        &= y_n^{-1} z_n^2
        + \brackets{\alpha_2 - y_n^2} y_n^{-1} z_n
        - \brackets{\alpha_1 + \alpha_2 - n} y_n,
    \end{align}
    then \,\, $y_n = (\ln \theta_{n + 1})' + t$ \,\, is a solution of the $\PIV[y_n; n]$ equation;

    \item[(b)]
    if the function 
    $$z_{m} = \eta_{m - 1} \, \eta_{m}^{-1} - (\alpha_1 + \alpha_2 - m)$$ 
    satisfies the equation
    \begin{align} \label{eq:zmcond}
        z_{m}' 
        &= y_{m - 1}^{-1} z_{m}^2
        + \brackets{\alpha_2 - y_{m - 1}^2} y_{m - 1}^{-1} z_{m}
        - \brackets{\alpha_1 + \alpha_2 - m} y_{m - 1},
    \end{align}
    then \,\, $y_{m - 1} = - (\ln \eta_{m - 1})' + t$ \,\, is a solution of the $\PIV[y_{m - 1}; m - 1]$ equation.
\end{itemize}
\end{thm}
\begin{proof}
Using the definitions of $y_n$, $z_n$ and conditions for $\theta_{n + 1}$ and $\eta_{m - 1}$, we can derive systems that are equivalent to the corresponding $\PIV$ equations (see Appendix \ref{app:proofthmscalthetanetamsol}).
\end{proof}

\begin{rem}
Note that for $n = m = 0$, the equations \eqref{eq:thetancond} and \eqref{eq:etamcond} take the form
\begin{align}
    \label{eq:theta1cond}
    &\theta_{1}''
    + t \, \theta_{1}'
    + 2 \theta_{1}^2 \, \theta_{0}^{-1}
    + \brackets{\alpha_0 - \alpha_1} \theta_{1} 
    = 0,
    \\[2mm]
    \label{eq:etam1cond}
    &\eta_{- 1}''
    - t \, \eta_{- 1}'
    + 2 \eta_{- 1}^2 \, \eta_0^{-1}
    + (\alpha_0 - \alpha_1 - 2) \eta_{- 1}
    = 0,
\end{align}
where $\theta_1 = \kappa_1$, $\eta_{-1} = \kappa_{-1}$ and $\theta_0 = \eta_{-1}^{-1}$, $\eta_0 = \theta_1^{-1}$. One can see that these equations coincide with \eqref{eq:kap1cond} and \eqref{eq:kapm1cond} from Proposition \ref{thm:ym1y0sol}, respectively. To obtain from these conditions the equations \eqref{eq:thetancond} and \eqref{eq:etamcond}, one can apply the $T_1^k$-operator to \eqref{eq:theta1cond} and \eqref{eq:etam1cond}, since for any $k \in \mathbb{Z}$ 
\begin{align}
    &&
    T_1^k (\alpha_0)
    &= \alpha_0 - k,
    &
    T_1^k (\alpha_1)
    &= \alpha_1 + k,
    &&
\end{align}
and
\begin{align}
    &&
    T_1^k (\theta_0)
    &= \theta_k,
    &
    T_1^{k} (\eta_{0})
    &= \eta_{k}.
    &&
\end{align}
The latter formulas are proved using the definition of $\tau_n$, namely $\tau_n = T_1^n (\tau_0)$, and mathematical induction.
\end{rem}

\begin{rem}
The additional condition \eqref{eq:kapm1kap1thirdcond} for $\kappa_{-1}$, $\kappa_1$ is rewritten in terms of the variables $\theta_{0}$, $\theta_{1}$ and $\eta_0$, $\eta_{- 1}$ as
\begin{align}
    \theta_0' \theta_1'
    + t \brackets{
    \theta_0' \theta_1 + \theta_0 \theta_1'
    }
    - \theta_1^2
    + \brackets{
    t^2 - \alpha_0 + \alpha_1 + 1
    } \theta_0 \theta_1
    &= \alpha_1 (1 - \alpha_0) \theta_0^2,
    \\[2mm]
    \eta_{-1}' \eta_0'
    - t \brackets{
    \eta_{-1}' \eta_0 + \eta_{-1} \eta_0'
    }
    - \eta_{-1}^2
    + \brackets{
    t^2 - \alpha_0 + \alpha_1 + 1
    } \eta_{-1} \eta_0
    &= \alpha_1 (1 - \alpha_0) \eta_0^2.
\end{align}
Applying the $T_1^{k}$-operator to these equalities, we obtain conditions
\begin{align}
    \label{eq:thetanthirdcond}
    &&
    \begin{aligned}
    &&
    \theta_n' \theta_{n + 1}'
    + t \brackets{
    \theta_n' \theta_{n + 1} + \theta_n \theta_{n + 1}'
    }
    - \theta_{n + 1}^2
    + \brackets{
    t^2 - \alpha_0 + \alpha_1 - 2 n + 1
    } \theta_{n} \theta_{n + 1}
    &&
    \\[1mm]
    &
    &= (\alpha_1 - n) (1 - \alpha_0 - n) \theta_n^2,
    &&
    \end{aligned}
    &&
    \begin{aligned}
    \\[1mm]
    n 
    & \geq 0,
    \end{aligned}
    &&
    \\[2mm]
    \label{eq:etamthirdcond}
    &&
    \begin{aligned}
    &&
    \eta_{m - 1}' \eta_{m}'
    - t \brackets{
    \eta_{m - 1}' \eta_{m} 
    + \eta_{m - 1} \eta_{m}'
    }
    - \eta_{m - 1}^2
    + \brackets{
    t^2 - \alpha_0 + \alpha_1 - 2 m + 1
    } \eta_{m - 1} \eta_{m}
    &&
    \\[1mm]
    &
    &= (\alpha_1 - m) (1 - \alpha_0 - m) \eta_{m}^2,
    &&
    \end{aligned}
    &&
    \begin{aligned}
    \\[1mm]
    m 
    &\leq 0,
    \end{aligned}
\end{align}
that coincide with those obtained by substitution of $y_n$, $z_n$ and $y_{m - 1}$, $z_{m}$ into \eqref{eq:zncond} and \eqref{eq:zmcond}, respectively.
\end{rem}

The following proposition connects solutions of the Toda chain \eqref{eq:scalToda} constructed from the Hankel matrix representation and those obtained by \Backlund transformations.

\begin{prop} \label{thm:scalarinitcompcond}
Solutions defined by \Backlund transformations and by Hankel matrices of the Toda equations are equivalent.
\end{prop}
\begin{proof}
Note that solutions of the Toda equations are uniquely determined by the initial conditions $\kappa_{-1}$ and $\kappa_1$. Thus the statement of this proposition follows from the fact that one can choose the same initial conditions so that $a_0 = \kappa_1$ and $b_0 = \kappa_{-1}$.
\end{proof}

The main result of this section is the following
\begin{thm} \label{thm:scalP4solmainthm}
Let $\kappa_n$, $n \in \mathbb{Z}$, be a function generated by the Toda chain \eqref{eq:kapTodan},
\begin{align}
    &&
    \brackets{
    \tfrac12 D_t^2
    + \kappa_{-1} \, \kappa_1
    } \kappa_{n} \cdot \kappa_n
    &= \kappa_{n - 1} \, \kappa_{n + 1},
    &
    \kappa_0
    &= 1,
    &&
\end{align}
and the initial conditions $\kappa_{-1}$, $\kappa_1$ satisfy the equations
\begin{gather}
    \begin{aligned}
        \kappa_{-1}'' - t \kappa_{-1}' + 2 \kappa_{-1}^2 \kappa_{1} + \brackets{\alpha_0 - \alpha_1 - 2} \kappa_{-1}
        &= 0,
        &&&&&&
        \kappa_{1}'' + t \kappa_{1}' + 2 \kappa_{-1} \kappa_{1}^2 + \brackets{\alpha_0 - \alpha_1} \kappa_1
        &= 0,
    \end{aligned}
    \\[1mm]
    \kappa_{-1}' \, \kappa_1'
    + t \brackets{
    \kappa_{-1}' \, \kappa_1
    - \kappa_{-1} \, \kappa_1'
    }
    + \kappa_{-1}^2 \kappa_1^2
    - (t^2 - \alpha_0 + \alpha_1 + 1) \kappa_{-1} \, \kappa_1
    = \alpha_1 (\alpha_0 - 1).
\end{gather}
Then the function $y_n = (\ln \kappa_{n + 1} \, \kappa_n^{-1})' + t$ is a solution of the $\PIV[y_n;n]$ equation:
\begin{align}
    y_n ''
    &= \tfrac12 y_n^{-1} (y_n')^2
    + \tfrac32 y_n^3 
    - 2 t y_n^2
    + \brackets{\tfrac12 t^2 + \alpha_0 - \alpha_1 + 2 n} y_n
    - \tfrac12 \alpha_2^2 y_n^{-1},
\end{align}
where $\alpha_0 + \alpha_1 + \alpha_2 = 1$.
\end{thm}
\begin{proof}
As we have already shown, this theorem can be proved by introducing new notation $\theta_n$, $n \geq 0$, and $\eta_m$,~$m \leq 0$, and applying the $T_1^k$-operator to the conditions for $\kappa_{-1}$, $\kappa_1$. It allows us to derive conditions for $\theta_n$ and $\eta_m$ that lead to the statement of the theorem.
\end{proof}

Since in the noncommutative case it is sometimes impossible to reduce a system of ODEs to a single ODE, we prefer to use the auxiliary function $z_n$ and a condition for it instead of \eqref{eq:thetanthirdcond} or \eqref{eq:etamthirdcond}.

Now we are going to reformulate Theorem \ref{thm:scalthetanetamsol} for the $\PIV[f_{i, n}; n]$ symmetric system
\begin{align}
    \label{eq:P4nsym}
    &\left\{
    \begin{array}{lcl}
         f_{0, n}'
         &=& f_{0, n} f_{1, n} 
         - f_{2, n} f_{0, n} 
         + (\alpha_0 + n),  
         \\[2mm]
         f_{1, n}'
         &=& f_{1, n} f_{2, n} 
         - f_{0, n} f_{1, n} 
         + (\alpha_1 - n),
         \\[2mm]
         f_{2, n}'
         &=& f_{2, n} f_{0, n} 
         - f_{1, n} f_{2, n} 
         + \alpha_2,
    \end{array}
    \right.
\end{align}
that can be obtained from \eqref{eq:P4sym} by applying the $T_1^n$-operator to it, by setting
\begin{align}
    T_1^n (f_0)
    &= f_{0, n},
    &
    T_1^n (f_1)
    &= f_{1, n},
    &
    T_1^n (f_2)
    &= f_{2, n}.
\end{align}

\begin{thm} \label{thm:P4nsymsol}
Let $\theta_{n}$, $n \geq 0$, and $\eta_{m}$, $m \leq 0$, be chosen as in Theorem \ref{thm:scalthetanetamsol}. Then
\begin{itemize}
    \item[(a)] the functions
    \begin{align}
        f_{0, n}
        &= - f_{1, n} - f_{2, n} + t,
        &
        \tilde f_{1, n}
        &= \theta_{n + 1} \, \theta_{n}^{-1}
        - (\alpha_1 - n)
        = f_{1, n} \, f_{2, n},
        &
        f_{2, n}
        &= (\ln \theta_{n + 1})' + t,
    \end{align}
    are solutions of the $\PIV[f_{i, n}; n]$ symmetric form \eqref{eq:P4nsym}, if $f_{1, n}$ satisfies the equation
    \begin{equation} \label{eq:f1neq}
        f_{1, n}'
        = f_{1, n}^2
        + 2 f_{1, n} f_{2, n} 
        - t f_{1, n}
        + (\alpha_1 - n);
    \end{equation}
    
    \item[(b)] the functions
    \begin{gather}
    \begin{aligned}
        f_{1, m - 1}
        &= - f_{0, m - 1} - f_{2, m - 1} + t,
    \end{aligned}
    \\[1mm]
    \begin{aligned}
        \tilde f_{0, m - 1}
        &= \eta_{m - 1} \, \eta_{m}^{-1}
        + (\alpha_0 + m - 1)
        = f_{0, m - 1} \, f_{2, m - 1},
        &&&
        f_{2, m - 1}
        &= - (\ln \eta_{m - 1})' + t,
    \end{aligned}
    \end{gather}
    are solutions of the $\PIV[f_{i, m - 1}; m - 1]$ symmetric form \eqref{eq:P4nsym}, if $f_{0, m - 1}$ satisfies the equation
    \begin{equation} \label{eq:f0neq}
        f_{0, m - 1}'
        = - f_{0, m - 1}^2
        - 2 f_{0, m - 1} f_{2, m - 1} 
        + t f_{0, m - 1}
        + (\alpha_0 + m - 1).
    \end{equation}
\end{itemize}
\end{thm}
\begin{proof}
We repeat the reasoning of the proof of Theorem \ref{thm:scalthetanetamsol}.

\medskip
\textbullet \,\, \textbf{Case (a).} 
By a straightforward computation, one can get the equation
\begin{align}
    f_{2, n}'
    &= - f_{2, n}^2 
    - 2 \tilde f_{1, n}
    + t f_{2, n}
    + \alpha_2
    = - f_{2, n}^2 
    - 2 f_{1, n} f_{2, n}
    + t f_{2, n}
    + \alpha_2.
\end{align}
Taking it together with \eqref{eq:f1neq}, we obtain the system
\begin{align}
    &\left\{
    \begin{array}{lcl}
         f_{1, n}'
         &=& f_{1, n}^2
         + 2 f_{1, n} f_{2, n} 
         - t f_{1, n}
         + (\alpha_1 - n),
         \\[2mm]
         f_{2, n}'
         &=& - f_{2, n}^2
         - 2 f_{1, n} f_{2, n} 
         + t f_{2, n}
         + \alpha_2.
    \end{array}
    \right.
\end{align}
Using the definition of $f_{0, n}$, it can be rewritten as
\begin{align}
    &\left\{
    \begin{array}{lcl}
         f_{1, n}'
         &=& f_{1, n} f_{2, n}
         - f_{0, n} f_{1, n}
         + (\alpha_1 - n),
         \\[2mm]
         f_{2, n}'
         &=& f_{2, n} f_{0, n}
         - f_{1, n} f_{2, n} 
         + \alpha_2,
    \end{array}
    \right.
\end{align}
where
\begin{align}
    f_{0, n}'
    &= \brackets{
    - f_{1, n} - f_{2, n} + t
    }'
    = - f_{1, n}' - f_{2, n}' + 1
    \\[1mm]
    &= - \brackets{
    f_{1, n} f_{2, n}
    - f_{0, n} f_{1, n}
    + (\alpha_1 - n)
    }
    - \brackets{
    f_{2, n} f_{0, n}
    - f_{1, n} f_{2, n} 
    + \alpha_2
    }
    + (\alpha_0 + \alpha_1 + \alpha_2)
    \\[1mm]
    &= f_{0, n} f_{1, n}
    - f_{2, n} f_{0, n}
    + (\alpha_0 + n).
\end{align}
So, we arrive at the $\PIV[f_{i, n}; n]$ symmetric form \eqref{eq:P4nsym}.

\medskip
\textbullet \,\, \textbf{Case (b).} 
In a similar way we obtain the equation
\begin{align}
    f_{2, m - 1}'
    &= f_{2, m - 1}^2 
    + 2 \tilde f_{0, m - 1}
    - t f_{2, m - 1}
    + \alpha_2
    = f_{2, m - 1}^2  
    + 2 f_{0, m - 1} f_{2, m - 1}
    - t f_{2, m - 1}
    + \alpha_2,
\end{align}
that, with the assumption \eqref{eq:f0neq} and the definition of $f_{1, m - 1}$, gives the $\PIV[f_{i, m - 1}; m - 1]$ symmetric form~\eqref{eq:P4nsym}.
\end{proof}

\section{Noncommutative version of the \Painleve IV system}
\label{sec:ncP4}

A noncommutative version of the $\PIV$ symmetric form \eqref{eq:P4sym}, that has the first integral 
\begin{align}
    I
    &= f_0 + f_1 + f_2 - t,
\end{align}
can be written as
\begin{align}
    &\left\{
    \begin{array}{lcl}
         f_{0}'
         &=& a_0 \, f_{0} f_{1}
         + (1 - a_0) \, f_1 f_0
         - a_2 \, f_{2} f_{0}
         - (1 - a_2) \, f_0 f_2
         + \alpha_0,  
         \\[2mm]
         f_{1}'
         &=& a_1 \, f_{1} f_{2}
         + (1 - a_1) \, f_2 f_1
         - a_0 \, f_{0} f_{1}
         - (1 - a_0) \, f_1 f_0
         + \alpha_1,
         \\[2mm]
         f_{2}'
         &=& a_2 \, f_{2} f_{0}
         + (1 - a_2) \, f_0 f_2
         - a_1 \, f_{1} f_{2}
         - (1 - a_1) \, f_2 f_1
         + \alpha_2,
    \end{array}
    \right.
\end{align}
where $\alpha_0 + \alpha_1 + \alpha_2 = 1$ and $a_i$ are arbitrary parameters. 
\begin{rem}
We remark that when $\alpha_0 + \alpha_1 + \alpha_2 = 0$, this system can be regarded as a fully noncommutative analog of \textit{the~Lotka-Volterra~system}.
\end{rem}

\begin{rem}
The non-abelian generalizaions for the Volterra lattices with central time and their \Painleve type reductions were studied in the paper \cite{adler2020}. In particular, the author have derived two non-equivalent \Painleve IV type systems that admit a \Backlund transformation different from those given in Table \ref{tab:BTonP4sym}.
\end{rem}

It turns out that \textit{this system admits the same \Backlund transformations as in the commutative case (see Table~\ref{tab:BTonP4sym}) with nonzero parameters $\alpha_i$ iff $a_0 = a_1 = a_2 = a$}. Then it takes the form
\begin{align}
    \label{eq:ncP4sym}
    &\left\{
    \begin{array}{lcl}
         f_{0}'
         &=& a \, f_{0} f_{1}
         + (1 - a) \, f_1 f_0
         - a \, f_{2} f_{0}
         - (1 - a) \, f_0 f_2
         + \alpha_0,  
         \\[2mm]
         f_{1}'
         &=& a \, f_{1} f_{2}
         + (1 - a) \, f_2 f_1
         - a \, f_{0} f_{1}
         - (1 - a) \, f_1 f_0
         + \alpha_1,
         \\[2mm]
         f_{2}'
         &=& a \, f_{2} f_{0}
         + (1 - a) \, f_0 f_2
         - a \, f_{1} f_{2}
         - (1 - a) \, f_2 f_1
         + \alpha_2.
    \end{array}
    \right.
\end{align}

\begin{rem} \label{rem:Tinv}
The involution $x \, y \overset{T}{\mapsto} y \, x$ preserves the form of system \eqref{eq:ncP4sym} changing the parameter $a$ to $1 - a$.
\end{rem}

\subsection{Solutions of the noncommutative \texorpdfstring{$\PIV$}{P4} system}
\label{sec:ncHankel}

In Section \ref{sec:scalHankel}, we have established that the commutative $\PIV$ equation possesses solutions expressible in terms of the Hankel determinant, if we impose some restrictions on the initial conditions $\kappa_{-1}$ and $\kappa_1$. In the current section, we are going to generalize this result to a fully noncommutative version of the $\PIV$ symmetric system of the form \eqref{eq:ncP4sym}. Since as we have remarked earlier, in noncommutative case it is sometimes impossible to reduce a system of ODEs to a single ODE (see, for instance, system \eqref{eq:ncP4sym}), we will generalize Theorem \ref{thm:P4nsymsol}.

Consider the functions $\theta_n$, $n \geq 0$, and $\eta_m$, $m \leq 0$ that are defined by the noncommutative Toda equations \eqref{eq:ncToda_thetan} and \eqref{eq:ncToda_etam}, respectively. Note that $\theta_0 = \eta_{-1}^{-1}$ and $\eta_0 = \theta_1^{-1}$.

\begin{prop} \label{thm:ncym1y0sol}
Let $\theta_0$, $\theta_{1}$ and $\eta_{0}$, $\eta_{-1}$ satisfy the conditions
\begin{align}
    \label{eq:nctheta1cond}
    &\theta_1''
    + t \, \theta_1'
    + 2 \theta_{1} \, \theta_0^{-1} \, \theta_1
    + (\alpha_0 - \alpha_1) \theta_1
    = 0,
    \\[1mm]
    \label{eq:ncetam1cond}
    &\eta_{-1}''
    - \eta_{-1}' \, t
    + 2 \eta_{-1} \, \eta_0^{-1} \, \eta_{-1}
    + (\alpha_0 - \alpha_1 - 2) \eta_{-1}
    = 0,
\end{align}
where $\theta_0 = \eta_{-1}^{-1}$ and $\eta_0 = \theta_{1}^{-1}$. 

Then
\begin{itemize}
    \item[(a)] the functions
    \begin{align}
        f_0
        &= - f_1 - f_2 + t,
        &
        \tilde f_1
        &= \theta_1 \, \theta_0^{-1} - \alpha_1
        = \tfrac12 f_1 f_2 + \tfrac12 f_2 f_1,
        &
        f_2
        &= \theta_1' \, \theta_1^{-1} + t,
    \end{align}
    are solutions of the $\PIV$ symmetric form \eqref{eq:ncP4sym} with $a = 1$, if $f_1$ satisfies the equation
    \begin{align}
        \label{eq:ncf1eq}
        f_1'
        &= f_1^2 + f_1 f_2 + f_2 f_1 - t f_1 + \alpha_1;
    \end{align}
    
    \item[(b)] the functions
    \begin{align}
        f_1
        &= - f_0 - f_2 + t,
        &
        \tilde f_0
        &= \eta_0^{-1} \, \eta_{- 1}
        + (\alpha_0 - 1)
        = \tfrac12 f_0 f_2 + \tfrac12 f_2 f_0,
        &
        f_2
        &= - \eta_{-1}^{-1} \, \eta_{-1}' + t,
    \end{align}
    are solutions of the $\PIV$ symmetric form \eqref{eq:ncP4sym} with $\tilde \alpha_0 = \alpha_0 - 1$, $\tilde \alpha_1 = \alpha_1 + 1$, and $a = 1$, if $f_0$ satisfies the equation
    \begin{align}
        \label{eq:ncf0eq}
        f_0'
        &= - f_0^2 - f_0 f_2 - f_2 f_0 + f_0 t + \tilde \alpha_0.
    \end{align}
\end{itemize}
\end{prop}

\begin{proof}
The proof is given by a straightforward computation, using definitions of $f_i$ and conditions~\eqref{eq:nctheta1cond}~--~\eqref{eq:ncetam1cond}.

\medskip
\textbullet \, \, \textbf{Case (a).} 
The derivative of $f_{2}$ can be written in the following way:
\begin{align}
    f_{2}'
    &= \theta_{1}'' \theta_{1}^{-1}
    - (\theta_{1}' \theta_{1}^{-1})^2
    + 1
    \\[1mm]
    &= - \brackets{
    t \theta_{1}'
    + 2 \theta_{1} \theta_0^{-1} \theta_{1}
    + (\alpha_0 - \alpha_1) \theta_{1}
    } \theta_{1}^{-1}
    - (\theta_{1}' \theta_{1}^{-1})^2
    + (\alpha_0 + \alpha_1 + \alpha_2)
    \\[1mm]
    &= - t \theta_{1}' \theta_{1}^{-1}
    - (\theta_{1}' \theta_{1}^{-1})^2
    - 2 \theta_{1} \theta_{0}^{-1}
    + (2 \alpha_1 + \alpha_2)
    \\[1mm]
    &= - t (f_{2} - t)
    - (f_{2} - t)^2
    - 2 \brackets{
    \tilde f_{1} + \alpha_1
    }
    + (2 \alpha_1 + \alpha_2)
    \\[1mm]
    &= - f_{2}^2 
    - 2 \tilde f_{1} + f_{2} t + \alpha_2
    = - f_{2}^2 
    - f_{2} f_{1} 
    - f_{1} f_{2} 
    + f_{2} t + \alpha_2.
\end{align}
Taking it together with the condition \eqref{eq:ncf1eq}, we obtain the system
\begin{align}
    \left\{
    \begin{array}{lcl}
         f_{1}'
         &=& f_{1}^2 
         + f_{1} f_{2}
         + f_{2} f_{1}
         - t f_{1}
         + \alpha_1,
         \\[2mm]
         f_{2}'
         &=& - f_{2}^2 
        - f_{2} f_{1} 
        - f_{1} f_{2} 
        + f_{2} t + \alpha_2.
    \end{array}
    \right.
\end{align}
Since $f_{0} = - f_{1} - f_{2} + t$, the system takes the form
\begin{align}
    &
    \left\{
    \begin{array}{lcl}
         f_{1}'
         &=& f_{1} f_{2}
         - \brackets{
         - f_{1} - f_{2} + t
         } f_{1}
         + \alpha_1,
         \\[2mm]
         f_{2}'
         &=& f_{2}
        \brackets{
        - f_{2}
        - f_{1}
        + t
        }
        - f_{1} f_{2} 
        + \alpha_2;
    \end{array}
    \right.
    &
    \Leftrightarrow&&
    &
    \left\{
    \begin{array}{lcl}
         f_{1}'
         &=& f_{1} f_{2}
         - f_{0} f_{1}
         + \alpha_1,
         \\[2mm]
         f_{2}'
         &=& f_{2} f_{0}
        - f_{1} f_{2} 
        + \alpha_2.
    \end{array}
    \right.
\end{align}
Finally, adding to this system the following equation for $f_{0}$:
\begin{align}
    f_{0}'
    &= \brackets{
    - f_{1} - f_{2} + t
    }' 
    = - f_{1}' - f_{2}' + 1
    \\[1mm]
    &= - \brackets{
    f_{1} f_{2}
    - f_{0} f_{1}
     + \alpha_1
    }
    - \brackets{
    f_{2} f_{0}
    - f_{1} f_{2} 
    + \alpha_2
    }
    + (\alpha_0 + \alpha_1 + \alpha_2)
    \\[1mm]
    &= f_{0} f_{1}
    - f_{2} f_{0}
    + \alpha_0,
\end{align}
we arrive at the $\PIV$ symmetric form \eqref{eq:ncP4sym} with $a = 1$.

\medskip
\textbullet \, \, \textbf{Case (b).} 
Similarly, for $f_{2}'$ we have:
\begin{align}
    f_{2}'
    &= (\eta_{- 1}^{-1} \eta_{- 1}')^2
    - \eta_{- 1}^{-1} \eta_{- 1}'' 
    + 1
    \\[1mm]
    &=  (\eta_{- 1}^{-1} \eta_{- 1}')^2
    + \eta_{- 1}^{-1} \brackets{
    - \eta_{- 1}' t
    + 2 \eta_{- 1} \eta_{0}^{-1} \eta_{- 1}
    + (\alpha_0 - \alpha_1 -2) 
    \eta_{- 1}
    }
    + (\alpha_0 + \alpha_1 + \alpha_2)
    \\[1mm]
    &= (\eta_{- 1}^{-1} \eta_{- 1}')^2
    - \eta_{- 1}^{-1} \eta_{- 1}' t
    + 2 \eta_{0}^{-1} \eta_{- 1}
    + (2 \alpha_0 + \alpha_2 - 2)
    \\[1mm]
    &= (- f_{2} + t)^2
    - (- f_{2} + t) t
    + 2 \brackets{
    \tilde f_{0} - (\alpha_0 - 1)
    }
    + (2 \alpha_0 + \alpha_2 - 2)
    \\[1mm]
    &= f_{2}^2 
    + 2 \tilde f_{0}
    - t f_{2} + \alpha_2
    = f_{2}^2 
    + f_{2} f_{0}
    + f_{0} f_{2}
    - t f_{2} + \alpha_2.
\end{align}
The equation and condition \eqref{eq:ncf0eq} give the system
\begin{align}
    \left\{
    \begin{array}{lcl}
         f_{0}'
         &=& - f_{0}^2 
         - f_{0} f_{2}
         - f_{2} f_{0}
         + f_{0} t
         + (\alpha_0 - 1),
         \\[2mm]
         f_{2}'
         &=& f_{2}^2 
        + f_{2} f_{0}
        + f_{0} f_{2}
        - t f_{2} + \alpha_2,
    \end{array}
    \right.
\end{align}
that, by the definition of $f_{1}$, can be rewritten as
\begin{align}
    \left\{
    \begin{array}{lcl}
         f_{0}'
         &=& f_{0} f_{1}
         - f_{2} f_{0}
         + (\alpha_0 - 1),
         \\[2mm]
         f_{2}'
         &=& f_{2} f_{0}
         - f_{1} f_{2}
         + \alpha_2.
    \end{array}
    \right.
\end{align}
Supplementing it by the equation for $f_{1}'$,
\begin{align}
    f_{1}'
    &= \brackets{
    - f_{0}
    - f_{2}
    + t
    }' 
    = - f_{0}'
    - f_{2}'
    + 1
    \\[1mm]
    &= - \brackets{
    f_{0} f_{1}
    - f_{2} f_{0}
    + (\alpha_0 - 1)
    }
    - \brackets{
    f_{2} f_{0}
    - f_{1} f_{2}
    + \alpha_2
    }
    + (\alpha_0 + \alpha_1 + \alpha_2)
    \\[1mm]
    &= f_{1} f_{2}
    - f_{0} f_{1}
    + (\alpha_1 + 1),
\end{align}
we get the $\PIV$ symmetric form \eqref{eq:ncP4sym} with $\tilde \alpha_0 = \alpha_0 - 1$, $\tilde \alpha_1 = \alpha_1 + 1$, and $a = 1$.
\end{proof}

\begin{rem}
Conditions \eqref{eq:nctheta1cond} -- \eqref{eq:ncetam1cond} are noncommutative analogs of conditions \eqref{eq:theta1cond} -- \eqref{eq:etam1cond}.
\end{rem}

Proposition \ref{thm:ncym1y0sol} means that there is only one system (up to the $T$-involution given in Remark \ref{rem:Tinv}) of the form \eqref{eq:ncP4sym}:
\begin{align}
    \label{eq:ncP4sym1}
    &\left\{
    \begin{array}{lcl}
         f_{0}'
         &=& f_{0} f_{1}
         - f_{2} f_{0}
         + \alpha_0,  
         \\[2mm]
         f_{1}'
         &=& f_{1} f_{2}
         - f_{0} f_{1}
         + \alpha_1,
         \\[2mm]
         f_{2}'
         &=& f_{2} f_{0}
         - f_{1} f_{2}
         + \alpha_2,
    \end{array}
    \right.
\end{align}
that has solutions in the quasideterminant Hankel form. This system is a generalization of the quantum $\PIV$ symmetric form defined in the paper \cite{nagoya2008quantum} and the matrix $\text{P}_4^0$ system from \cite{Bobrova_Sokolov_2021_1} to the fully noncommutative case. In the paper \cite{Bobrova_Sokolov_2021_2}, the authors have suggested a noncommutative version of the $\PIV$ system with the noncommutative independent variable. This system admits solutions in the quasideterminant form only in the "positive" or "negative" directions for $k = 2$ or $k = 0$. The noncommutative version \eqref{eq:ncP4sym1} of the $\PIV$ symmetric form allows us to write solutions in both directions. 

Now we are going to define solutions in the "positive" and "negative" directions. Since system \eqref{eq:ncP4sym1} has the same Weyl group as in the commutative case, we are able to define the $T_1$-operator. Applying it to \eqref{eq:ncP4sym1}, we obtain the noncommutative $\PIV[f_{i, n}; n]$ symmetric form
\begin{align}
    \label{eq:ncP4nsym}
    &\left\{
    \begin{array}{lcl}
         f_{0, n}'
         &=& f_{0, n} f_{1, n}
         - f_{2, n} f_{0, n}
         + (\alpha_0 + n),  
         \\[2mm]
         f_{1, n}'
         &=& f_{1, n} f_{2, n}
         - f_{0, n} f_{1, n}
         + (\alpha_1 - n),
         \\[2mm]
         f_{2, n}'
         &=& f_{2, n} f_{0, n}
         - f_{1, n} f_{2, n}
         + \alpha_2.
    \end{array}
    \right.
\end{align}
where $\alpha_0 + \alpha_1 + \alpha_2 = 1$ and $T_1^{n} (f_i) = f_{i, n}$. Note that the system has the following first integral
\begin{equation}
    I =
    f_{0, n} + f_{1, n} + f_{2, n} - t.
\end{equation}

\begin{thm} \label{thm:ncP4nsymsol}
Let the functions $\theta_{n}$, $n \geq 0$ and $\eta_{m}$, $m \leq 0$
satisfy the noncommutative Toda equations \eqref{eq:ncToda_thetan} -- \eqref{eq:ncToda_etam} and the following equations
\begin{align}
    \label{eq:ncthetancond}
    &\theta_{n + 1}''
    + t \, \theta_{n + 1}'
    + 2 \theta_{n + 1} \, \theta_n^{-1} \, \theta_{n + 1}
    + \brackets{\alpha_0 - \alpha_1 + 2 n} \theta_{n + 1} 
    = 0,
    \\[2mm]
    \label{eq:ncetamcond}
    &\eta_{m - 1}''
    - \eta_{m - 1}' \, t
    + 2 \eta_{m - 1} \, \eta_{m}^{-1} \, \eta_{m - 1}
    + (\alpha_0 - \alpha_1 + 2 (m - 1)) \eta_{m - 1}
    = 0.
\end{align}
Then
\begin{itemize}
    \item[(a)] the functions
    \begin{align}
        f_{0, n}
        &= - f_{1, n} - f_{2, n} + t,
        &
        \tilde f_{1, n}
        &= \theta_{n + 1} \, \theta_{n}^{-1}
        - (\alpha_1 - n)
        = \tfrac12 f_{1, n} \, f_{2, n}
        + \tfrac12 f_{2, n} \, f_{1, n}
        ,
        &
        f_{2, n}
        &= \theta_{n + 1}' \, \theta_{n + 1}^{-1} + t,
    \end{align}
    are solutions of the $\PIV[f_{i, n}; n]$ symmetric form \eqref{eq:ncP4nsym}, if $f_{1, n}$ satisfies the equation
    \begin{equation} \label{eq:ncf1neq}
        f_{1, n}'
        = f_{1, n}^2
        + f_{1, n} f_{2, n} 
        + f_{2, n} f_{1, n} 
        - t f_{1, n}
        + (\alpha_1 - n);
    \end{equation}
    
    \item[(b)] the functions
    \begin{gather}
    \begin{aligned}
        f_{1, m - 1}
        &= - f_{0, m - 1} - f_{2, m - 1} + t,
    \end{aligned}
    \\[1mm]
    \begin{aligned}
        \tilde f_{0, m - 1}
        &= \eta_{m}^{-1} \, \eta_{m - 1}
        + (\alpha_0 + m - 1)
        = \tfrac12 f_{0, m - 1} \, f_{2, m - 1}
        + \tfrac12 f_{2, m - 1} \, f_{0, m - 1}
        ,
        &&&
        f_{2, m - 1}
        &= - \eta_{m - 1}^{-1} \, \eta_{m - 1}' + t,
    \end{aligned}
    \end{gather}
    are solutions of the $\PIV[f_{i, m - 1}; m - 1]$ symmetric form \eqref{eq:ncP4nsym}, if $f_{0, m - 1}$ satisfies the equation
    \begin{equation} \label{eq:ncf0neq}
        f_{0, m - 1}'
        = - f_{0, m - 1}^2
        - f_{0, m - 1} f_{2, m - 1} 
        - f_{2, m - 1} f_{0, m - 1} 
        + f_{0, m - 1} t
        + (\alpha_0 + m - 1).
    \end{equation}
\end{itemize}
\end{thm}
\begin{proof}
Using the conditions \eqref{eq:ncthetancond} -- \eqref{eq:ncetamcond} and definitions of $f_{i, n}$, one is able to derive the corresponding $\PIV$ symmetric form \eqref{eq:ncP4nsym} by the same computations as in Proposition \ref{thm:ncym1y0sol}. For more details see Appendix~\ref{app:proofthmcnP4nsymsol}.
\end{proof}

\begin{rem}
Conditions \eqref{eq:ncthetancond} and \eqref{eq:ncetamcond} are derived from \eqref{eq:nctheta1cond} and \eqref{eq:ncetam1cond}, respectively, by acting the $T_1^k$-operator on them.
\end{rem}

The next proposition follows from the fact that in the noncommutative case \Backlund transformations preserve the Hankel property just like in the commutative case (see Proposition \ref{thm:scalarinitcompcond}).
\begin{prop}
Solutions defined by \Backlund transformations and by Hankel matrices of the noncommutative Toda equations are equivalent.
\end{prop}
\begin{proof}
Follows from the same arguments as in Proposition \ref{thm:scalarinitcompcond}.
\end{proof}

Summarising results of this section, we formulate the following
\begin{thm} \label{thm:ncP4solmainthm}
Let the functions $\theta_{n}$, $n \geq 0$ and $\eta_{m}$, $m \leq 0$
be defined by the noncommutative Toda chains \eqref{eq:ncToda_thetan} -- \eqref{eq:ncToda_etam} and initial conditions $\theta_0$, $\theta_1$ and $\eta_0$, $\eta_{-1}$ satisfy the equations
\begin{align}
    \label{eq:nctheta1cond1}
    &\theta_1''
    + t \, \theta_1'
    + 2 \theta_{1} \, \theta_0^{-1} \, \theta_1
    + (\alpha_0 - \alpha_1) \theta_1
    = 0,
    \\[1mm]
    \label{eq:ncetam1cond1}
    &\eta_{-1}''
    - \eta_{-1}' \, t
    + 2 \eta_{-1} \, \eta_0^{-1} \, \eta_{-1}
    + (\alpha_0 - \alpha_1 - 2) \eta_{-1}
    = 0,
\end{align}
where $\theta_0 = \eta_{-1}^{-1}$ and $\eta_0 = \theta_{1}^{-1}$. 

Then
\begin{itemize}
    \item[(a)] the functions
    \begin{align}
        f_{0, n}
        &= - f_{1, n} - f_{2, n} + t,
        &
        \tilde f_{1, n}
        &= \theta_{n + 1} \, \theta_{n}^{-1}
        - (\alpha_1 - n)
        = \tfrac12 f_{1, n} \, f_{2, n}
        + \tfrac12 f_{2, n} \, f_{1, n}
        ,
        &
        f_{2, n}
        &= \theta_{n + 1}' \, \theta_{n + 1}^{-1} + t,
    \end{align}
    are solutions of the $\PIV[f_{i, n}; n]$ symmetric form \eqref{eq:ncP4nsym}, if $f_{1, n}$ satisfies the equation
    \begin{equation} \label{eq:ncf1neq1}
        f_{1, n}'
        = f_{1, n}^2
        + f_{1, n} f_{2, n} 
        + f_{2, n} f_{1, n} 
        - t f_{1, n}
        + (\alpha_1 - n);
    \end{equation}
    
    \item[(b)] the functions
    \begin{gather}
    \begin{aligned}
        f_{1, m - 1}
        &= - f_{0, m - 1} - f_{2, m - 1} + t,
    \end{aligned}
    \\[1mm]
    \begin{aligned}
        \tilde f_{0, m - 1}
        &= \eta_{m}^{-1} \, \eta_{m - 1}
        + (\alpha_0 + m - 1)
        = \tfrac12 f_{0, m - 1} \, f_{2, m - 1}
        + \tfrac12 f_{2, m - 1} \, f_{0, m - 1}
        ,
        &&&
        f_{2, m - 1}
        &= - \eta_{m - 1}^{-1} \, \eta_{m - 1}' + t,
    \end{aligned}
    \end{gather}
    are solutions of the $\PIV[f_{i, m - 1}; m - 1]$ symmetric form \eqref{eq:ncP4nsym}, if $f_{0, m - 1}$ satisfies the equation
    \begin{equation} \label{eq:ncf0neq1}
        f_{0, m - 1}'
        = - f_{0, m - 1}^2
        - f_{0, m - 1} f_{2, m - 1} 
        - f_{2, m - 1} f_{0, m - 1} 
        + f_{0, m - 1} t
        + (\alpha_0 + m - 1).
    \end{equation}
\end{itemize}
\end{thm}

\subsection{"Hamiltonian" structure and the Lax representation}

\subsubsection{"Hamiltonian" structure}
\label{sec:ncham}

The Hamiltonian for the commutative $\PIV$ symmetric form is given by \eqref{eq:scalh0}:
\begin{align}
    \label{eq:scalH}
    H (f_0, f_1, f_2)
    &= f_0 f_1 f_2
    + \tfrac13 (\alpha_1 - \alpha_2) f_0
    + \tfrac13 (\alpha_1 + 2 \alpha_2) f_1
    - \tfrac13 (2 \alpha_1 + \alpha_2) f_2.
\end{align}
Using a skew-symmetric matrix $U$,
\begin{align}
    &&
    U
    &= 
    \begin{pmatrix}
    0 & 1 & -1
    \\
    -1 & 0 & 1
    \\
    1 & -1 & 0
    \end{pmatrix}
    = ||u_{i j}||,
    &
    i, j
    &= 0, 1, 2,
    &&
\end{align}
one can define a Poisson bracket as
\begin{align}
    \label{eq:scalbr}
    &&
    \{ f_i, f_j \}
    &= u_{i, j}.
    &&
\end{align}
It turns out that the commutative $\PIV$ symmetric form is represented in terms of the linear Poisson bracket \eqref{eq:scalbr} with Hamiltonian \eqref{eq:scalH} in the following form
\begin{align}
    \left\{
    \begin{array}{lcl}
         f_0'
         &=& \{ H, f_0 \} + 1,
         \\[1mm]
         f_1'
         &=& \{ H, f_1 \},
         \\[1mm]
         f_2'
         &=& \{ H, f_2 \},
    \end{array}
    \right.
\end{align}
and is Hamiltonian iff $\alpha_0 + \alpha_1 + \alpha_2 = 0$.
On the other hand, if we introduce the canonical variables
\begin{align}
    q
    &:= f_1,
    &
    p
    &:= f_2,
    &
    t:= f_0 + f_1 + f_2,
\end{align}
then the Hamiltonian takes the form
\begin{align}
    H (q, p, t)
    &= - p^2 q 
    - p q^2 
    + p q t
    - \alpha_1 p
    + \alpha_2 q 
    + \tfrac13 (\alpha_1 - \alpha_2) t
\end{align}
and gives the system
\begin{align}
    \left\{
    \begin{array}{lclclcl}
         q'
         &=& \{ H, q \}
         &=& - \partial_p H
         &=& q^2 + 2 p q - q t + \alpha_1,
         \\[1mm]
         p'
         &=& \{ H, p \}
         &=& \partial_q H
         &=& - p^2 - 2 p q + p t + \alpha_2.
    \end{array}
    \right.
\end{align}

In the noncommutative case we have a similar structures. 
\begin{prop}
Let us introduce the following "canonical" variables
\begin{align}
    q
    &:= f_1,
    &
    p
    &:= f_2,
    &
    t
    &:= f_0 + f_1 + f_2,
\end{align}
where functions $f_i$, $i = 0, 1, 2$ satisfy the fully noncommutative $\PIV$ symmetric form. In terms of $p$ and $q$ the function
\begin{gather}
\label{eq:ncHfi}
\begin{multlined}
    H (f_0, f_1, f_2)
    = a_0 f_0 f_1 f_2 
    + (2 - a_0 - a_1) f_1 f_2 f_0
    + a_1 f_2 f_0 f_1
    \hspace{5cm}
    \\
    - (1 - a_1) f_1 f_0 f_2
    + (1 - a_0 - a_1) f_0 f_2 f_1
    - (1 - a_0) f_2 f_1 f_0
    \\
    + \tfrac13 (\alpha_1 - \alpha_2) f_0
    + \tfrac13 (\alpha_1 + 2 \alpha_2) f_1
    - \tfrac13 (2 \alpha_1 + \alpha_2) f_2
\end{multlined}
\end{gather}
obtains the form
\begin{multline}
    H (q, p, t)
    = - (1 - a_0) p^2 q 
    + (1 - 2 a_0) p q p 
    - (1 - a_0) q p^2 
    + (1 - a_0 - a_1) q^2 p 
    - (3 - 2 a_0 - 2 a_1) q p q 
    + (1 - a_0 - a_1) p q^2
    \\
    + (1 - a_0 - a_1) t p q 
    + a_1 p t q 
    - (1 - a_0) p q t 
    + (2 - a_0 - a_1) q p t 
    - (1 - a_1) q t p 
    + a_0 t q p
    \\
    - \alpha_1 p 
    + \alpha_2 q 
    + \tfrac13 (\alpha_1 - \alpha_2) t
\end{multline}
and the $\PIV$ system is equivalent to the following "Hamiltonian" system
\begin{align}
    \left\{
    \begin{array}{lclcl}
         q'
         &=& - \partial_p H
         &=& q^2 
         + q p + p q 
         - t q 
         + \alpha_1,
         \\[2mm]
         p'
         &=& \partial_q H
         &=& - p^2
         - p q - q p
         + p t
         + \alpha_2.
    \end{array}
    \right.
\end{align}
\end{prop}
\begin{proof}
By a straightforward computation.
\end{proof}

\begin{rem}
When $a_0 = a_1 = 1$, the "Hamiltonian" reads
\begin{align}
    H
    &= f_0 f_1 f_2 
    + f_2 f_0 f_1
    - f_0 f_2 f_1
    + \tfrac13 (\alpha_1 - \alpha_2) f_0
    + \tfrac13 (\alpha_1 + 2 \alpha_2) f_1
    - \tfrac13 (2 \alpha_1 + \alpha_2) f_2
    \\[1mm]
    &= - p q p
    - q^2 p
    + q p q
    - p q^2
    - t p q
    + p t q
    + t q p
    - \alpha_1 p 
    + \alpha_2 q 
    + \tfrac13 (\alpha_1 - \alpha_2) t
\end{align}
and is equivalent to a matrix Hamiltonian found in \cite{Kawakami_2015}.
\end{rem}

\begin{rem}
Since the ``Poisson brackets'' are defined by the matrix $U$, one can consider their naive ``quantization'', which is the replacement of the Poisson brackets by commutators of the generators $f_i$. This can always be done thanks to Theorem 1.2 from the paper \cite{Farkas1998RingTF}. A ``quantized Hamiltonian'' takes the form 
\begin{equation}
    H 
    = f_2 f_0 f_1
    + \tfrac13 (\alpha_1 - \alpha_2 + 3 \lambda) f_0
    + \tfrac13 (\alpha_1 + 2 \alpha_2) f_1
    - \tfrac13 (2 \alpha_1 + \alpha_2) f_2,
\end{equation}
where $\PoissonBrackets{f_i, f_j} = \lambda^{-1} \LieBrackets{{f_i, f_j}}$, $\lambda \in F$.
\end{rem}

\subsubsection{Lax representation}
\label{sec:isompair}

In this section we will discuss different isomonodromic Lax pairs for the noncommutative $\PIV[f_{i, n}; n]$ symmetric form \eqref{eq:ncP4nsym}.

\begin{prop}
For any $n \in \mathbb{Z}$, system \eqref{eq:ncP4nsym} can be written as the linear system
\begin{align}
    \left\{
    \begin{array}{lcl}
         \partial_{\lambda} \Psi_n (\lambda, t)
         &=& \mathcal{A}_n (\lambda, t) \, \Psi_n (\lambda, t),
         \\[2mm]
         \partial_{t} \Psi_n (\lambda, t)
         &=& \mathcal{B}_n (\lambda, t) \, \Psi_n (\lambda, t),
    \end{array}
    \right.
\end{align}
where matrices $\mathcal{A}_n (\lambda, t)$ and $\mathcal{B}_n (\lambda, t)$ have the following $\lambda$-dependence
\begin{align} \label{eq:ncNYpair}
    &&
    \mathcal{A}_n (\lambda, t)
    &= A_0 + A_{-1} \lambda^{-1},
    &
    \mathcal{B}_n (\lambda, t)
    &= B_1 \lambda + B_0,
    &&
\end{align}
with matrices $A_0$, $A_{-1}$, $B_1$, and $B_0$ are given by
\begin{align}
    A_0
    &= 
    \begin{pmatrix}
    0 & 1 & f_{0, n}
    \\[0.5mm]
    0 & 0 & 1
    \\[0.5mm]
    0 & 0 & 0
    \end{pmatrix}
    ,
    &
    A_{-1}
    &=
    \begin{pmatrix}
    \beta_0 & 0 & 0
    \\[0.5mm]
    f_{1, n} & \beta_1 & 0
    \\[0.5mm]
    1 & f_{2, n} & \beta_2
    \end{pmatrix}
    ,
    &
    B_1
    &= 
    \begin{pmatrix}
    0 & 0 & 1 
    \\[0.5mm]
    0 & 0 & 0 
    \\[0.5mm]
    0 & 0 & 0
    \end{pmatrix}
    ,
    &
    B_0
    &= 
    \begin{pmatrix}
    - f_{2, n} & 0 & 0
    \\[0.5mm]
    1 & - f_{0, n} & 0
    \\[0.5mm]
    0 & 1 & - f_{1, n}
    \end{pmatrix}
    .
\end{align}
Here scalar parameters $\beta_0$, $\beta_1$, $\beta_2$ are related to the $\alpha$'s parameters as
\begin{align} \label{eq:albetrel}
    \alpha_0
    &= 1 + \beta_2 - \beta_0 - n,
    &
    \alpha_1
    &= \beta_0 - \beta_1 + n,
    &
    \alpha_2 
    &= \beta_1 - \beta_2.
\end{align}
\end{prop}
\begin{proof}
The compatibility condition
\begin{align}
    \partial_t \mathcal{A}_n 
    - \partial_{\lambda} \mathcal{B}_n
    &= \LieBrackets{\mathcal{B}_n, \mathcal{A}_n}
\end{align}
leads to the system
\begin{align}
    \left\{
    \begin{array}{lcl}
         f_{0, n}'
         &=& f_{0, n} f_{1, n}
         - f_{2, n} f_{0, n}
         + (1 + \beta_2 - \beta_0),
         \\[2mm]
         f_{1, n}'
         &=& f_{1, n} f_{2, n}
         - f_{0, n} f_{1, n}
         + (\beta_0 - \beta_1),
         \\[2mm]
         f_{2, n}'
         &=& f_{2, n} f_{0, n} 
         - f_{1, n} f_{2, n}
         + (\beta_1 - \beta_2)
    \end{array}
    \right.
\end{align}
that is exactly the $\PIV[f_{i, n}; n]$ symmetric form \eqref{eq:ncP4nsym}, since relations \eqref{eq:albetrel} hold.
\end{proof}

\begin{rem}
The pair is constructed by using the method of non-abelianization of the well-known scalar pairs, suggested in the paper \cite{Bobrova_Sokolov_2021_2}.
\end{rem}

Note that in the commutative case the Lax pair $\mathcal{A}_0(\lambda, t)$, $\mathcal{B}_0(\lambda, t)$ is equivalent to the Noumi-Yamada pair \cite{noumi2000affine} for the $\PIV$ symmetric form \eqref{eq:P4sym}.
The pair is reduced to the Jimbo-Miwa pair \cite{Jimbo_Miwa_1981} that has the following dependence on the spectral parameter $\mu$
\begin{align}
    \mathbf{A}_0 (\mu, t)
    &= A_1 \mu
    + A_0
    + A_{-1} \mu^{-1},
    &
    \mathbf{B}_0 (\mu, t)
    &= B_1 \mu
    + B_0,
\end{align}
where $A_{1}$, $A_0$, $A_{-1}$, $B_1$, $B_0$ are $2 \times 2$ matrices and $A_1$, $B_1$ are constant diagonal matrices with different eigenvalues. The reduction is given by a generalized Laplace transform and some constraints on the parameters $\beta_i$ \cite{joshi2007linearization}. One can establish the same fact in the noncommutative case.

\begin{prop}
The Noumi-Yamada pair \eqref{eq:ncNYpair} can be reduced to a pair of the Jimbo-Miwa type.
\end{prop}
\begin{proof}
We will repeat the proof of this fact in the commutative case, presented in the paper \cite{joshi2007linearization}.

\medskip

\textbullet \, \, \textbf{Laplace transform.} First we formally consider the following Laplace transform
\begin{align}
    \Psi_n (\lambda, t)
    &= \int_C e^{\lambda \mu} \, \Phi_n (\mu, t) \, d \mu,
\end{align}
where we assume that one can choose a contour $C$ such that the corresponding terms arising from integration by parts will cancel out. 
\begin{rem}
Strict conditions of the existence such a transformation will be explored in forthcoming articles.
\end{rem}
This transformation turns the linear system
\begin{align}
    \left\{
    \begin{array}{lcl}
         \partial_{\lambda} \Psi_n (\lambda, t)
         &=& (A_0 + A_{-1} \lambda^{-1}) \, \Psi_n (\lambda, t),
         \\[2mm]
         \partial_{t} \Psi_n (\lambda, t)
         &=& (B_1 \lambda + B_0) \, \Psi_n (\lambda, t)
    \end{array}
    \right.
\end{align}
into the following linear system
\begin{align}
    \left\{
    \begin{array}{rclcl}
         \partial_{\mu} \Phi_n (\mu, t)
         &=& (A_0 - \mu \, \mathbb{I})^{-1} (A_{-1} + \mathbb{I}) \, \Phi_n (\mu, t)
         &=& \mathbf{A}_n (\mu, t) \, \Phi_n (\mu, t)
         ,
         \\[2mm]
         \partial_{t} \Phi_n (\mu, t)
         &=& - B_1 \, \partial_{\mu} \Phi_n (\mu, t)
         + B_0 \Phi_n (\mu, t)
         &=& \mathbf{B}_n (\mu, t) \, \Phi_n (\mu, t)
         ,
    \end{array}
    \right.
\end{align}
where matrices $\mathbf{A}_n$ and $\mathbf{B}_n$ have the form
\begin{align}
    \mathbf{A}_n (\mu, t)
    &= A_{-1} \, \mu^{-1}
    + A_{-2} \, \mu^{-2}
    + A_{-3} \, \mu^{-3},
    &
    \mathbf{B}_n (\mu, t)
    &= B_0 
    + B_{-1} \, \mu^{-1}.
\end{align}
Changing the spectral parameter as $\mu \mapsto \mu^{-1}$, we obtain the pair
\begin{align}
    \mathbf{A}_n (\mu, t)
    &= A_{1} \, \mu
    + A_{0}
    + A_{-1} \, \mu^{-1},
    &
    \mathbf{B}_n (\mu, t)
    &= B_1 \, \mu
    + B_{0},
\end{align}
where matrices $A_1$, $A_0$, $A_{-1}$, $B_1$, and $B_0$ are
\begin{gather}
\begin{aligned}
    A_1
    &= 
    \begin{pmatrix}
    1 & f_{2, n} & \beta_2 + 1
    \\[0.5mm]
    0 & 0 & 0
    \\[0.5mm]
    0 & 0 & 0
    \end{pmatrix}
    ,
    &
    A_0
    &= 
    \begin{pmatrix}
    f_{0, n} + f_{1, n} 
    & f_{0, n} f_{2, n} + (\beta_1 + 1)
    & (\beta_2 + 1) f_{0, n}
    \\[0.5mm]
    1 & f_{2, n} & \beta_2 + 1
    \\[0.5mm]
    0 & 0 & 0
    \end{pmatrix}
    ,
\end{aligned}
\\[2mm]
\begin{aligned}
    A_{-1}
    &= 
    \begin{pmatrix}
    \beta_0 + 1 & 0 & 0
    \\[0.5mm]
    f_{1, n} & \beta_1 + 1 & 0 
    \\[0.5mm]
    1 & f_{2, n} & \beta_2 + 1
    \end{pmatrix}
    ,
    &
    B_1
    &= 
    \begin{pmatrix}
    1 & f_{2, n} & \beta_2 + 1
    \\[0.5mm]
    0 & 0 & 0 
    \\[0.5mm]
    0 & 0 & 0
    \end{pmatrix}
    ,
    &
    B_0
    &= 
    \begin{pmatrix}
    - f_{2, n} & 0 & 0
    \\[0.5mm]
    1 & - f_{0, n} & 0
    \\[0.5mm]
    0 & 1 & - f_{1, n}
    \end{pmatrix}
    .
\end{aligned}
\end{gather}

\medskip

\textbullet \,\, 
\textbf{Reduction to a $2 \times 2$ pair.} 
Let us present the fundamental solution $\Phi_n (\mu, t)$ as a column vector of the form $\Phi_n = \begin{pmatrix} \Phi_{1, n} & \Phi_{2, n} & \Phi_{3, n}\end{pmatrix}^T$. Without loss of generality, one can set $\beta_2 = -1$. In view of this condition, we define the column vector $\tilde \Phi_n = \begin{pmatrix} \Phi_{1, n} & \Phi_{2, n}\end{pmatrix}^T$ that satisfies the following linear system
\begin{align}
    \left\{
    \begin{array}{rclcl}
         \partial_{\mu} \tilde \Phi_n (\mu, t)
         &=& \mathbf{\tilde A}_n (\mu, t) \, \tilde \Phi_n (\mu, t)
         ,
         \\[2mm]
         \partial_{t} \tilde \Phi_n (\mu, t)
         &=& \mathbf{\tilde B}_n (\mu, t) \, \tilde \Phi_n (\mu, t)
         ,
    \end{array}
    \right.
\end{align}
with
\begin{gather}
    \mathbf{\tilde A}_n (\mu, t)
    = 
    \begin{pmatrix}
    1 & f_{2, n}
    \\[0.5mm]
    0 & 0
    \end{pmatrix}
    \mu
    + 
    \begin{pmatrix}
    f_{0, n} + f_{1, n} 
    & f_{0, n} f_{2, n} + (\beta_1 + 1)
    \\[0.5mm]
    1 & f_{2, n}
    \end{pmatrix}
    + 
    \begin{pmatrix}
    \beta_0 + 1 & 0
    \\[0.5mm]
    f_{1, n} & \beta_1 + 1
    \end{pmatrix}
    \mu^{-1},
    \\[2mm]
    \mathbf{\tilde B}_n (\mu, t)
    = 
    \begin{pmatrix}
    1 & f_{2, n}
    \\[0.5mm]
    0 & 0
    \end{pmatrix}
    \mu
    + 
    \begin{pmatrix}
    - f_{2, n} 
    & 0
    \\[0.5mm]
    1 & - f_{0, n}
    \end{pmatrix}.
\end{gather}
Making a gauge transformation by the matrix $g = \begin{pmatrix} 1 & f_{2, n} \\[0.5mm] 0 & 1 \end{pmatrix}$, the pair is written in the Jimbo-Miwa form:
\begin{gather}
\label{eq:ncJMpair}
\begin{gathered}
    \begin{multlined}
    \mathbf{\tilde A}_n (\mu, t)
    = 
    \begin{pmatrix}
    1 & 0
    \\[0.5mm]
    0 & 0
    \end{pmatrix}
    \mu
    + 
    \begin{pmatrix}
    f_{0, n} + f_{1, n} + f_{2, n}
    & - f_{1, n} f_{2, n} + (\beta_1 + 1)
    \\[0.5mm]
    1 & 0
    \end{pmatrix}
    \hspace{2.72cm}
    \\[1mm]
    + 
    \begin{pmatrix}
    f_{2, n} f_{1, n} + (\beta_0 + 1) & - f_{2, n} f_{1, n} f_{2, n} - (\beta_0 - \beta_1) f_2
    \\[0.5mm]
    f_{1, n} & - f_{1, n} f_{2, n} + (\beta_1 + 1)
    \end{pmatrix}
    \mu^{-1},
    \end{multlined}
    \\[2mm]
    \mathbf{\tilde B}_n (\mu, t)
    = 
    \begin{pmatrix}
    1 & 0
    \\[0.5mm]
    0 & 0
    \end{pmatrix}
    \mu
    + 
    \begin{pmatrix}
    0 & - f_{1, n} f_{2, n} + (\beta_1 + 1)
    \\[0.5mm]
    1 & - f_{0, n} - f_{2, n}
    \end{pmatrix}.
\end{gathered}
\end{gather}
The compatibility condition
\begin{align}
    \partial_t \mathbf{\tilde A}_n
    - \partial_{\mu} \mathbf{\tilde B}_n
    &= \LieBrackets{\mathbf{\tilde B}_n, \mathbf{\tilde A}_n}
\end{align}
leads to the system that is equivalent to the noncommutative $\PIV[f_{i, n}; n]$ symmetric form \eqref{eq:ncP4nsym}.
\end{proof}

\begin{rem}
If we use the condition $f_{0, n} + f_{1, n} + f_{2, n} = t$ with an appropriate choice of index $n$ in  the Jimbo-Miwa pair \eqref{eq:ncJMpair}, we obtain pairs for systems \eqref{eq:ncf1f2nsys} and \eqref{eq:ncf0f2mm1sys}.
In particular, the substitution of $f_{0, n} = - f_{1, n} - f_{2, n} + t$ into \eqref{eq:ncJMpair} leads to a pair that is a particular case of a slightly general pair written in Conclusion of the paper \cite{Bobrova_Sokolov_2021_2}.
\end{rem}

\section*{Conclusion}

Using \Backlund transformations and the symmetric form of the commutative $\PIV$ equation, we construct its fully noncommutative version that possesses solutions in terms of the infinite Toda system. Our approach can be useful for constructing such analogs for other \Painleve equations. Their connection to the solutions of the noncommutative Toda system can be applied to describing Darboux-\Backlund transformations of their solutions. 

Our analog leads to a fully noncommutative hierarchy of the $\PIV$ systems and to fully noncommutative analogs of the systems of type $A_l^{(1)}$, $l = 2, 3, \dots, N$. One of such infinite sequence of the $\PIV$ systems is suggested in \cite{gordoa2021matrix} in the matrix case. The second type of the hierarchy that contains the commutative $\PIV$ symmetric form was introduced in \cite{veselov1993} as a dressing chain and was investigated in the paper \cite{noumi2000affine} with the help of the affine Weyl group symmetries.

In the commutative case, the generating functions for the entries of the Hankel determinants are connected with the asymptotic solution at infinity of the isomonodromic problem for the $\PPII$ and $\PIV$ equations \cite{joshi2004generating}, \cite{joshi2006generating}. It would be interesting to study this problem in a fully noncommutative case.

\appendix

\section{}
\label{sec:app}
\subsection{Proof of Theorem \ref{thm:scalthetanetamsol}}
\label{app:proofthmscalthetanetamsol}
\phantom{}

\medskip
\textbullet \,\, \textbf{Case (a).} 
Take the derivative of $y_n$ w.r.t. $t$ and use condition \eqref{eq:thetancond}:
\begin{align}
    y_n'
    &= \brackets{\theta_{n + 1}' \theta_{n + 1}^{-1} + t}'
    = \theta_{n + 1}'' \theta_{n + 1}^{- 1} 
    - \brackets{\theta_{n + 1}' \theta_{n + 1}^{-1}}^2
    + 1
    \\[1mm]
    &= -\brackets{
    t \theta_{n + 1}' 
    + 2 \theta_{n + 1}^2 \theta_n^{-1}
    + (\alpha_0 - \alpha_1 + 2 n) \theta_{n + 1}
    } \theta_{n + 1}^{- 1} 
    - \brackets{\theta_{n + 1}' \theta_{n + 1}^{-1}}^2
    + 1
    \\[1mm]
    &= - t \theta_{n + 1}' \theta_{n + 1}^{-1} 
    - 2 \theta_{n + 1} \theta_n^{-1}
    - (\alpha_0 - \alpha_1 + 2 n)
    - (\theta_{n + 1}' \theta_{n + 1}^{-1})^2
    + 1.
\end{align}
Replace $\theta_{n + 1}' \theta_{n + 1}^{-1}$ by $y_n - t$ and $\theta_{n + 1} \theta_n^{-1}$ by $z_n + (\alpha_1 + \alpha_2 - n)$ and use $\alpha_0 + \alpha_1 + \alpha_2 = 1$:
\begin{align}
    y_n'
    &= - t (y_n - t) 
    - 2 \brackets{
    z_n + (\alpha_1 + \alpha_2 - n)
    }
    - (\alpha_0 - \alpha_1 + 2 n)
    - (y_n - t)^2
    + 1
    \\[1mm]
    &= - y_n^2
    - 2 z_n + t y_n - \alpha_2.
\end{align}
Hence, we obtain the system
\begin{align} \label{eq:P4nsys}
    \left\{
    \begin{array}{lcl}
         - z_n'
         &=& y_n^{-1} z_n^2 
         + (\alpha_2 - y_n^2) y_n^{-1} z_n
         - (\alpha_1 + \alpha_2 - n) y_n,
         \\[2mm]
         - y_n'
         &=& y_n^2 
         + 2 z_n 
         - t y_n 
         + \alpha_2,
    \end{array}
    \right.
\end{align}
that is equivalent to the $\PIV[y_n; n]$ equation \eqref{eq:P4yn}.

\medskip
\textbullet \,\, \textbf{Case (b).} 
There is the same chain of identities:
\begin{align}
    y_{m - 1}'
    &= \brackets{- \eta_{m - 1}' \eta_{m - 1}^{-1} + t}'
    = - \eta_{m - 1}'' \eta_{m - 1}^{- 1} 
    + \brackets{\eta_{m - 1}' \eta_{m - 1}^{-1}}^2
    + 1
    \\[1mm]
    &= \brackets{
    - t \, \eta_{m - 1}'
    + 2 \eta_{m - 1}^2 \, \eta_{m}^{-1}
    + (\alpha_0 - \alpha_1 + 2 (m - 1)) \eta_{m - 1}
    } \eta_{m - 1}^{- 1} 
    + \brackets{\eta_{m - 1}' \eta_{m - 1}^{-1}}^2
    + 1
    \\[1mm]
    &= - t \eta_{m - 1}' \eta_{m - 1}^{-1} 
    + 2 \eta_{m - 1} \eta_{m}^{-1}
    + \brackets{\eta_{m - 1}' \eta_{m - 1}^{-1}}^2
    + (\alpha_0 - \alpha_1 + 2 m - 1).
\end{align}
Replace $\eta_{m - 1}' \eta_{m - 1}^{-1}$ by $- y_{m - 1} + t$ and $\eta_{m - 1} \eta_{m}^{-1}$ by $z_{m} + (\alpha_1 + \alpha_2 - m)$ and use $\alpha_0 + \alpha_1 + \alpha_2 = 1$:
\begin{align}
    y_{m - 1}'
    &= - t (- y_{m - 1} + t) 
    + 2 \brackets{
    z_{m} + (\alpha_1 + \alpha_2 - m)
    }
    + (- y_{m - 1} + t)^2
    + (- 2 \alpha_1 - \alpha_2 + 2 m)
    \\[1mm]
    &= y_{m - 1}^2
    + 2 z_{m} - t y_{m - 1} + \alpha_2.
\end{align}
So, the resulting system
\begin{align} \label{eq:P4msys}
    \left\{
    \begin{array}{lcl}
         z_{m}'
         &=& y_{m - 1}^{-1} z_{m}^2 
         + (\alpha_2 - y_{m - 1}^2) y_{m - 1}^{-1} z_{m}
         - (\alpha_1 + \alpha_2 - m) y_{m - 1},
         \\[2mm]
         y_{m - 1}'
         &=& y_{m - 1}^2 
         + 2 z_{m} 
         - t y_{m - 1} 
         + \alpha_2,
    \end{array}
    \right.
\end{align}
reduces to the $\PIV[y_{m - 1}; m - 1]$ equation \eqref{eq:P4yn} for $y_{m - 1} = y_{m - 1} (t)$.

\subsection{Proof of Theorem \ref{thm:ncP4nsymsol}}
\label{app:proofthmcnP4nsymsol}
\phantom{}

\textbullet \, \, \textbf{Case (a).} 
The derivative of $f_{2, n}$:
\begin{align}
    f_{2, n}'
    &= \theta_{n + 1}'' \theta_{n + 1}^{-1}
    - (\theta_{n + 1}' \theta_{n + 1}^{-1})^2
    + 1
    \\[1mm]
    &= - \brackets{
    t \theta_{n + 1}'
    + 2 \theta_{n + 1} \theta_n^{-1} \theta_{n + 1}
    + (\alpha_0 - \alpha_1 + 2 n) \theta_{n + 1}
    } \theta_{n + 1}^{-1}
    - (\theta_{n + 1}' \theta_{n + 1}^{-1})^2
    + (\alpha_0 + \alpha_1 + \alpha_2)
    \\[1mm]
    &= - t \theta_{n + 1}' \theta_{n + 1}^{-1}
    - (\theta_{n + 1}' \theta_{n + 1}^{-1})^2
    - 2 \theta_{n + 1} \theta_{n}^{-1}
    + (2 \alpha_1 + \alpha_2 - 2 n)
    \\[1mm]
    &= - t (f_{2, n} - t)
    - (f_{2, n} - t)^2
    - 2 \brackets{
    \tilde f_{1, n} + (\alpha_1 - n)
    }
    + (2 \alpha_1 + \alpha_2 - 2 n)
    \\[1mm]
    &= - f_{2, n}^2 - 2 \tilde f_{1, n} + f_{2, n} t + \alpha_2
    = - f_{2, n}^2 
    - f_{2, n} f_{1, n} 
    - f_{1, n} f_{2, n} 
    + f_{2, n} t + \alpha_2.
\end{align}
This equation together with \eqref{eq:ncf1neq} give the system
\begin{align}
\label{eq:ncf1f2nsys}
    \left\{
    \begin{array}{lcl}
         f_{1, n}'
         &=& f_{1, n}^2 
         + f_{1, n} f_{2, n}
         + f_{2, n} f_{1, n}
         - t f_{1, n}
         + (\alpha_1 - n),
         \\[2mm]
         f_{2, n}'
         &=& - f_{2, n}^2 
        - f_{2, n} f_{1, n} 
        - f_{1, n} f_{2, n} 
        + f_{2, n} t + \alpha_2.
    \end{array}
    \right.
\end{align}
Since $f_{0, n} = - f_{1, n} - f_{2, n} + t$, it can be rewritten as
\begin{align}
    &
    \left\{
    \begin{array}{lcl}
         f_{1, n}'
         &=& f_{1, n} f_{2, n}
         - f_{0, n} f_{1, n}
         + (\alpha_1 - n),
         \\[2mm]
         f_{2, n}'
         &=& f_{2, n} f_{0, n}
        - f_{1, n} f_{2, n} 
        + \alpha_2.
    \end{array}
    \right.
\end{align}
Supplementing the system with the following equation for $f_{0, n}$:
\begin{align}
    f_{0, n}'
    &= \brackets{
    - f_{1, n} - f_{2, n} + t
    }' 
    = - f_{1, n}' - f_{2, n}' + 1
    \\[1mm]
    &= - \brackets{
    f_{1, n} f_{2, n}
    - f_{0, n} f_{1, n}
     + (\alpha_1 - n)
    }
    - \brackets{
    f_{2, n} f_{0, n}
    - f_{1, n} f_{2, n} 
    + \alpha_2
    }
    + (\alpha_0 + \alpha_1 + \alpha_2)
    \\[1mm]
    &= f_{0, n} f_{1, n}
    - f_{2, n} f_{0, n}
    + (\alpha_0 + n),
\end{align}
we get the $\PIV[f_{i, n}; n]$ symmetric form \eqref{eq:ncP4nsym}.

\medskip
\textbullet \, \, \textbf{Case (b).} 
Similarly, we have the following chain of identities:
\vspace{-0.5cm}
\begin{align}
    f_{2, m - 1}'
    &= (\eta_{m - 1}^{-1} \eta_{m - 1}')^2
    - \eta_{m - 1}^{-1} \eta_{m - 1}'' 
    + 1
    \\[1mm]
    &=  (\eta_{m - 1}^{-1} \eta_{m - 1}')^2
    + \eta_{m - 1}^{-1} \brackets{
    - \eta_{m - 1}' t
    + 2 \eta_{m - 1} \eta_{m}^{-1} \eta_{m - 1}
    + (\alpha_0 - \alpha_1 + 2 (m - 1)) 
    \eta_{m - 1}
    }
    + (\alpha_0 + \alpha_1 + \alpha_2)
    \\[1mm]
    &= (\eta_{m - 1}^{-1} \eta_{m - 1}')^2
    - \eta_{m - 1}^{-1} \eta_{m - 1}' t
    + 2 \eta_{m}^{-1} \eta_{m - 1}
    + (2 \alpha_0 + \alpha_2 + 2 (m - 1))
    \\[1mm]
    &= (- f_{2, m - 1} + t)^2
    - (- f_{2, m - 1} + t) t
    + 2 \brackets{
    \tilde f_{0, m - 1} - (\alpha_0 + (m - 1))
    }
    + (2 \alpha_0 + \alpha_2 + 2 (m - 1))
    \\[1mm]
    &= f_{2, m - 1}^2 
    + 2 \tilde f_{0, m - 1}
    - t f_{2, m - 1} + \alpha_2
    = f_{2, m - 1}^2 
    + f_{2, m - 1} f_{0, m - 1}
    + f_{0, m - 1} f_{2, m - 1}
    - t f_{2, m - 1} + \alpha_2.
\end{align}
So we arrive at the system
\begin{align}
\label{eq:ncf0f2mm1sys}
    \left\{
    \begin{array}{lcl}
         f_{0, m - 1}'
         &=& - f_{0, m - 1}^2 
         - f_{0, m - 1} f_{2, m - 1}
         - f_{2, m - 1} f_{0, m - 1}
         + f_{0, m - 1} t
         + (\alpha_0 + (m - 1)),
         \\[2mm]
         f_{2, m - 1}'
         &=& f_{2, m - 1}^2 
        + f_{2, m - 1} f_{0, m - 1}
        + f_{0, m - 1} f_{2, m - 1}
        - t f_{2, m - 1} + \alpha_2,
    \end{array}
    \right.
\end{align}
that by the definition of $f_{1, m - 1}$, can be represented as
\begin{align}
    \left\{
    \begin{array}{lcl}
         f_{0, m - 1}'
         &=& f_{0, m - 1} f_{1, m - 1}
         - f_{2, m - 1} f_{0, m - 1}
         + (\alpha_0 + (m - 1)),
         \\[2mm]
         f_{2, m - 1}'
         &=& f_{2, m - 1} f_{0, m - 1}
         - f_{1, m - 1} f_{2, m - 1}
         + \alpha_2.
    \end{array}
    \right.
\end{align}
Taking it with the equation for $f_{1, m - 1}'$,
\begin{align}
    f_{1, m - 1}'
    &= \brackets{
    - f_{0, m - 1}
    - f_{2, m - 1}
    + t
    }' 
    = - f_{0, m - 1}'
    - f_{2, m - 1}'
    + 1
    \\[1mm]
    &= - \brackets{
    f_{0, m - 1} f_{1, m - 1}
    - f_{2, m - 1} f_{0, m - 1}
    + (\alpha_0 + (m - 1))
    }
    \\
    & \qquad
    - \brackets{
    f_{2, m - 1} f_{0, m - 1}
    - f_{1, m - 1} f_{2, m - 1}
    + \alpha_2
    }
    + (\alpha_0 + \alpha_1 + \alpha_2)
    \\[1mm]
    &= f_{1, m - 1} f_{2, m - 1}
    - f_{0, m - 1} f_{1, m - 1}
    + (\alpha_1 - (m - 1)),
\end{align}
one can get the $\PIV[f_{i, m - 1}; m - 1]$ symmetric form \eqref{eq:ncP4nsym}.

\section{}
\label{sec:VV}
\par
{\rm 1.  Sokolov V. V. and  Svinolupov S. I.}, {\it
Vector-matrix generalizations of classical integrable equations}, Teoret.
and Mat. Fiz, 1994, {\bf 100}, no.{\bf 2},  214-218 \quad [in Russian];
 \quad {\it translation in Theoret. and Math. Phys.}, 1994, {\bf 100},
no.{\bf 2},  
\href{https://link.springer.com/article/10.1007/BF01016758}{959 - 962} (1994). 
\quad 35Qxx (35A30)
\medskip
\par
{\rm 2.   Olver P. J. and   Sokolov V. V.}, {\it
Integrable evolution equations on associative algebras},
Comm. in Math. Phys., 1998, {\bf 193}, no.{\bf 2},  
\href{https://link.springer.com/article/10.1007/s002200050328}{245-268}.
\quad 58F07 (35Q51)
\medskip
\par
{\rm 3.   Balandin S. P. and  Sokolov V. V.}, {\it On the Painlev\'e test
for non-Abelian equations}, Phys. Lett. A, 1998, {\bf 246}, no.{\bf 3-4},
\href{https://www.sciencedirect.com/science/article/abs/pii/S0375960198003363?via\%3Dihub}{267-272}.  
\quad 34A34 (34A25)
\medskip
\par
{\rm 4.   Olver P. J. and   Sokolov V. V.}, {\it Non-abelian integrable
systems of the derivative nonlinear Schr\"odinger type},
Inverse Problems, 1998, {\bf 14}, no.{\bf 6},
\href{https://iopscience.iop.org/article/10.1088/0266-5611/14/6/002}{L5-L8}. 
\quad 35Q55 (58F37)
\medskip
\par
{\rm 5.  Mikhailov A. V. and  Sokolov V. V.}, {\it Integrable ODEs on Associative Algebras},
Comm. in Math. Phys., 2000, {\bf 211}, no.{\bf 1}, 
\href{https://link.springer.com/article/10.1007/s002200050810}{231-251}.
\medskip
\par
{\rm 6. Odesskii A. V. and Sokolov V. V.}, 
{\it Integrable matrix equations related to pairs of compatible associative algebras},
Journal Phys. A: Math. Gen., 2006, {\bf 39},
\href{https://iopscience.iop.org/article/10.1088/0305-4470/39/40/011/meta}{12447-12456}.
\quad
\href{https://arxiv.org/abs/math/0604574}{arXiv:math/0604574}
\medskip
\par
{\rm 7. Odesskii A. V., Roubtsov V. N., and Sokolov V. V.},
{\it Bi-Hamiltonian ODEs with matrix variables},  
Teoret. and Mat. Fiz. 2012, {\bf 171}, no.{\bf 1}, 
\href{http://www.mathnet.ru/php/archive.phtml?wshow=paper&jrnid=tmf&paperid=6912&option_lang=eng}{26--33}.  
\quad 
\href{https://arxiv.org/abs/1105.1740}{arXiv:1105.1740}.
\medskip
\par
{\rm 8. Odesskii A. V., Roubtsov V. N., and Sokolov V. V.},
{\it Double Poisson brackets on free associative algebras}, 
Max-Planck-Institut fur Mathematik Preprint Series 2012, no.{\bf 52}, 19 p.; 
\quad    
\href{https://arxiv.org/abs/1208.2935}{arXiv:1208.2935}; 
\quad Contemporary Mathematics,
2013,    {\bf 592}, \href{http://dx.doi.org/10.1090/conm/592/11861}{225-241}, 
ISBN-10: 0-8218-8980-X,
ISBN-13: 978-0-8218-8980-0
\medskip
\par
{\rm 9. Sokolov V.~V. and Wolf T.},
{\it  Non-commutative generalization of integrable quadratic ODE-systems},    
{Letters in Mathematical Physics}, 
2020, {\bf 110},
no.{\bf 3}, 
\href{https://link.springer.com/article/10.1007/s11005-019-01229-0}{533-553}.  
\quad  
\href{https://arxiv.org/abs/1807.05583}{arXiv:1807.05583}
\medskip
\par
{\rm 10.   Sokolov V.~V.}, 
{\it Nonabelian $\mathfrak{so}_3$ Euler top}, {Russian Math. Surveys.},  
2021, {\bf 76}(1), 
\href{http://dx.doi.org/10.4213\%2Frm9988}{195-196}.
 
\href{https://arxiv.org/abs/2101.00934}{arXiv:2101.00934}
\medskip
\par  
{\rm 11.  Adler V.~E. and Sokolov V.~V.}, 
{\it Non-Abelian evolution systems with conservation laws}, 
{Mathematical Physics Analysis and Geometry},
2021, 
{\bf 24}, 
no.{\bf 7}, 
\href{https://link.springer.com/article/10.1007/s11040-021-09382-6}{1-24}.
\quad  
\href{https://arxiv.org/abs/2008.09174}{arXiv:2008.09174}
\medskip
\par 
{\rm 12.  Adler V.~E. and Sokolov V.~V.},
{\it On matrix Painlev\'e II equations}, 
{Theor. and Math. Phys.}, 2021, 
{\bf 207}(2),  
\href{https://doi.org/10.4213/tmf10027}{188--201}.
\quad   
\href{https://arxiv.org/abs/2012.05639}{arXiv:2012.05639}
\medskip
\par
{\rm 13.  
Odesskii A. V. and Sokolov V.~V.}, {\it Noncommutative elliptic Poisson structures on projective spaces}, {Journal of Geometry and Physics}, 2021, {\bf 169}, 
\href{https://doi.org/10.1016/j.geomphys.2021.104330}{104330}. 
\quad 
\href{https://arxiv.org/abs/1911.03320}{arXiv:1911.03320}
\medskip
\par
{\rm 14. Bobrova I.~A. and Sokolov V.~V.},
{\it On matrix Painlev\'e-4 equations. Part 1: Painlev\'e--Kovalevskaya test},  
2021,  20 pp. 
\quad 
\href{https://arxiv.org/abs/2107.11680}{arXiv:2107.11680} 
\medskip
\par 
{\rm 15. Bobrova I.~A. and Sokolov V.~V.},
{\it On matrix Painlev\'e-4 equations. Part 2: Isomonodromic Lax pairs},  
2021, 12 pp. 
\quad 
\href{https://arxiv.org/abs/2110.12159}{arXiv:2110.12159}
\medskip
\par  
{\rm 16.  Bobrova I.~A. and Sokolov V.~V.}, {\it On matrix Painlev\'e-4 equations,} submitted to {\em Nonlinearity},  
extended version is presented in 
\quad 
\href{https://arxiv.org/abs/2107.11680}{arXiv:2107.11680}  
\, and \, 
\href{https://arxiv.org/abs/2110.12159}{arXiv:2110.12159}
\medskip
\par   
{\rm 17.  Bobrova I.~A. and Sokolov V.~V.},
{\it Non-abelian Painlev\'e systems with~generalized~Okamoto~integral},

in preparation
\medskip

    \bibliographystyle{alpha}
    \bibliography{bib}

\newcommand{\etalchar}[1]{$^{#1}$}
\begin{thebibliography}{KMN{\etalchar{+}}99}

\bibitem[Adl21]{adler2020}
V.~E. Adler.
\newblock Painlev{\'e} type reductions for the non-{A}belian {V}olterra
  lattices.
\newblock {\em Journal of Physics A: Mathematical and Theoretical},
  54(3):\href{https://doi.org/10.1088/1751--8121/abd21f}{035204}, 2021.
\newblock \href{https://arxiv.org/abs/2010.09021}{arXiv:2010.09021}.

\bibitem[AS21]{Adler_Sokolov_2020_1}
V.~E. Adler and V.~V. Sokolov.
\newblock On matrix {P}ainlev{\'e} {II} equations.
\newblock {\em Theoret. and Math. Phys.},
  207(2):\href{https://doi.org/10.4213/tmf10027}{188--201}, 2021.
\newblock \href{https://arxiv.org/abs/2012.05639}{arXiv:2012.05639}.

\bibitem[BS98]{Balandin_Sokolov_1998}
S.~P. Balandin and V.~V. Sokolov.
\newblock On the {P}ainlev{\'e} test for non-{A}belian equations.
\newblock {\em Physics letters A},
  246(3-4):\href{https://www.sciencedirect.com/science/article/abs/pii/S0375960198003363?via\%3Dihub}{267--272},
  1998.

\bibitem[BS21a]{Bobrova_Sokolov_2021_1}
I.~A. Bobrova and V.~V. Sokolov.
\newblock On matrix {P}ainlev{\'e}-4 equations. {P}art 1:
  {P}ainlev{\'e}-{K}ovalevskaya test.
\newblock {\em arXiv preprint
  \href{https://arxiv.org/abs/2107.11680}{arXiv:2107.11680}}, 2021.

\bibitem[BS21b]{Bobrova_Sokolov_2021_2}
I.~A. Bobrova and V.~V. Sokolov.
\newblock On matrix {P}ainlev{\'e}-4 equations. {P}art 2: {I}somonodromic {L}ax
  pairs.
\newblock {\em arXiv preprint
  \href{https://arxiv.org/abs/2110.12159}{arXiv:2110.12159}}, 2021.

\bibitem[CdlI14]{cafasso2014non}
M.~Cafasso and M.~D. de~la Iglesia.
\newblock Non-commutative {P}ainlev{\'e} equations and {H}ermite-type matrix
  orthogonal polynomials.
\newblock {\em Communications in Mathematical Physics},
  326(2):\href{https://doi.org/10.1007/s00220--013--1853--4}{559--583}, 2014.
\newblock \href{https://arxiv.org/abs/1301.2116}{arXiv:1301.2116}.

\bibitem[EGR97]{etingof1997factorization}
P.~Etingof, I.~Gelfand, and V.~Retakh.
\newblock Factorization of differential operators, quasideterminants, and
  nonabelian {T}oda field equations.
\newblock {\em Math. Res. Lett.}, 4:413–25, 1997.
\newblock \href{https://arxiv.org/abs/q-alg/9701008}{arXiv:q-alg/9701008}.

\bibitem[FL98]{Farkas1998RingTF}
D.~R. Farkas and G.~Letzter.
\newblock Ring theory from symplectic geometry.
\newblock {\em Journal of Pure and Applied Algebra},
  125:\href{https://www.sciencedirect.com/science/article/pii/S002240499600117X?via\%3Dihub}{155--190},
  1998.

\bibitem[GP21]{gordoa2021matrix}
P.~R. Gordoa and A.~Pickering.
\newblock On matrix fourth {P}ainlev{\'e} hierarchies.
\newblock {\em Journal of Differential Equations},
  271:\href{https://www.sciencedirect.com/science/article/abs/pii/S0022039620304575}{499--532},
  2021.

\bibitem[GR91]{gel1991determinants}
I.~M. Gelfand and V.~S. Retakh.
\newblock Determinants of matrices over noncommutative rings.
\newblock {\em Functional Analysis and Its Applications},
  25(2):\href{https://link.springer.com/article/10.1007/BF01079588}{91--102},
  1991.

\bibitem[GR92]{gel1992theory}
I.~M. Gelfand and V.~S. Retakh.
\newblock A theory of noncommutative determinants and characteristic functions
  of graphs.
\newblock {\em Functional Analysis and Its Applications},
  26(4):\href{https://link.springer.com/article/10.1007/BF01075044}{231--246},
  1992.

\bibitem[JKM04]{joshi2004generating}
N.~Joshi, K.~Kajiwara, and M.~Mazzocco.
\newblock Generating function associated with the determinant formula for the
  solutions of the {P}ainlev{\'e} {II} equation.
\newblock {\em Aust{\'e}risque}, 297:67 -- 78, 2004.
\newblock \href{https://arxiv.org/abs/nlin/0406035}{arXiv:nlin/0406035}.

\bibitem[JKM06]{joshi2006generating}
N.~Joshi, K.~Kajiwara, and M.~Mazzocco.
\newblock Generating function associated with the {H}ankel determinant formula
  for the solutions of the {P}ainlev{\'e} {IV} equation.
\newblock {\em Funkcialaj Ekvacioj}, 49(3):451--468, 2006.
\newblock \href{https://arxiv.org/abs/nlin/0512041v3}{arXiv:nlin/0512041}.

\bibitem[JKT07]{joshi2007linearization}
N.~Joshi, A.~V. Kitaev, and P.~A. Treharne.
\newblock On the linearization of the {P}ainlev{\'e} {III--VI} equations and
  reductions of the three-wave resonant system.
\newblock {\em Journal of Mathematical Physics},
  48(10):\href{https://aip.scitation.org/doi/10.1063/1.2794560}{103512}, 2007.
\newblock \href{https://arxiv.org/abs/0706.1750v3}{arXiv:0706.1750v3}.

\bibitem[JM81]{Jimbo_Miwa_1981}
M.~Jimbo and T.~Miwa.
\newblock Monodromy preserving deformation of linear ordinary differential
  equations with rational coefficients. {II}.
\newblock {\em Physica D: Nonlinear Phenomena},
  2(3):\href{https://www.sciencedirect.com/science/article/abs/pii/016727898190021X}{407--448},
  1981.

\bibitem[Kaw15]{Kawakami_2015}
H.~Kawakami.
\newblock Matrix {P}ainlev{\'e} systems.
\newblock {\em Journal of Mathematical Physics},
  56(3):\href{https://doi.org/10.1063/1.4914369}{033503}, 2015.

\bibitem[KMN{\etalchar{+}}99]{kajiwara1999determinant}
K.~Kajiwara, T.~Masuda, M.~Noumi, Y.~Ohta, and Y.~Yamada.
\newblock Determinant formulas for the {T}oda and discrete {T}oda equations.
\newblock {\em arXiv preprint
  \href{https://arxiv.org/abs/solv-int/9908007}{solv-int/9908007}}, 1999.

\bibitem[Kri81]{krichever1981periodic}
I.~M. Krichever.
\newblock The periodic non-abelian {T}oda chain and its two-dimensional
  generalization.
\newblock {\em Russ. Math. Surv}, 36(2):82--89, 1981.

\bibitem[Mik81]{mikhailov1981reduction}
A.~V. Mikhailov.
\newblock The reduction problem and the inverse scattering method.
\newblock {\em Physica D: Nonlinear Phenomena},
  3(1-2):\href{https://doi.org/10.1016/0167--2789(81)90120--2}{73--117}, 1981.

\bibitem[Nag04]{nagoya2004quantum}
H.~Nagoya.
\newblock Quantum {P}ainlev{\'e} systems of type {$A_l^{(1)}$}.
\newblock {\em International Journal of Mathematics},
  15(10):\href{https://doi.org/10.1142/S0129167X0400265X}{1007--1031}, 2004.
\newblock \href{https://arxiv.org/abs/math/0402281v2}{arXiv:math/0402281v2 }.

\bibitem[NGR{\etalchar{+}}08]{nagoya2008quantum}
H.~Nagoya, B.~Grammaticos, A.~Ramani, et~al.
\newblock Quantum {P}ainlev{\'e} equations: from {C}ontinuous to discrete.
\newblock {\em SIGMA. Symmetry, Integrability and Geometry: Methods and
  Applications}, 4:\href{https://www.emis.de/journals/SIGMA/2008/051/}{051},
  2008.

\bibitem[Nou04]{noumi2004painleve}
M.~Noumi.
\newblock {\em Painlev{\'e} Equations Through Symmetry}, volume 223.
\newblock Springer Science \& Business, 2004.

\bibitem[NY00]{noumi2000affine}
M.~Noumi and Y.~Yamada.
\newblock Affine {W}eyl group symmetries in {P}ainlev{\'e} type equations.
\newblock {\em Towards the exact WKB analysis of differential equations, linear
  or nonlinear}, pages
  \href{http://citeseerx.ist.psu.edu/viewdoc/summary?doi=10.1.1.29.6416}{245--–259},
  2000.

\bibitem[Oka81]{okamoto1981tau}
K.~Okamoto.
\newblock On the $\tau$-function of the {P}ainlev{\'e} equations.
\newblock {\em Physica D: Nonlinear Phenomena},
  2(3):\href{https://www.sciencedirect.com/science/article/abs/pii/0167278981900269}{525--535},
  1981.

\bibitem[RR10]{Retakh_Rubtsov_2010}
V.~S. Retakh and V.~N. Rubtsov.
\newblock Noncommutative {T}oda {C}hains, {H}ankel {Q}uasideterminants and
  {P}ainlev{\'e} {II} {E}quation.
\newblock {\em Journal of Physics A: Mathematical and Theoretical},
  43(50):\href{https://iopscience.iop.org/article/10.1088/1751--8113/43/50/505204}{505204},
  2010.
\newblock \href{https://arxiv.org/abs/1007.4168}{arXiv:1007.4168}.

\bibitem[VS93]{veselov1993}
A.~P. Veselov and A.~B. Shabat.
\newblock Dressing chains and the spectral theory of the {S}chr{\"o}dinger
  operator.
\newblock {\em Functional Analysis and Its Applications},
  27(2):\href{https://doi.org/10.1007/BF01085979}{81--96}, 1993.

\end{thebibliography}
\end{document}